\documentclass[12pt,reqno]{amsart}
\usepackage[margin=1.2in]{geometry}

\usepackage{hyperref}
\author[ ]{}
\usepackage{cleveref}
\date{\vspace{-5ex}}
\usepackage{amsthm, amsfonts, mathrsfs}
\usepackage{ dsfont }
\usepackage{amssymb}
\usepackage{enumerate}
\usepackage{graphicx}
\usepackage{float}
\usepackage{bbm}
\usepackage{amsmath}
\usepackage{comment}
\usepackage{hyperref}
\usepackage{listings}
\usepackage{color}
\usepackage[dvipsnames]{xcolor}
\usepackage[lined,boxed,commentsnumbered]{algorithm2e}
\usepackage[colorinlistoftodos]{todonotes}
\usepackage[dvipsnames]{xcolor}
\allowdisplaybreaks

\hypersetup{
    colorlinks,
    citecolor=blue,
    filecolor=blue,
    linkcolor=blue,
    urlcolor=blue
}

\newtheorem{theorem}{Theorem}[section]
\newtheorem{proposition}[theorem]{Proposition}
\newtheorem{lemma}[theorem]{Lemma}

\newtheorem{corollary}[theorem]{Corollary}
\newtheorem{definition}[theorem]{Definition}
\newtheorem{remark}[theorem]{Remark}

\newtheorem{example}{Example}

\newcommand{\R}{\mathbb{R}}

\newcommand{\RE}{\overline {\R}}

\newcommand{\supp}{{\rm supp\,}}

\newcommand {\ND} {\mathcal{N}_{\mu_0}}
\newcommand{\TMP} {\mathbb T}
\newcommand{\BM} {\mathcal{M}}
\newcommand{\SM} {\mathcal{SM}}
\newcommand{\inv}{\dagger}
\newcommand{\STPM}{\mathcal {T}_{\mu_{0}}} 

\newcommand {\bpf} {\begin {proof}}
\newcommand {\epf} {\end {proof}}

\definecolor{Ivcolor}{RGB}{0, 153, 0}

\usepackage[
backend=biber,
style=numeric,
sorting=ynt
]{biblatex}
\usepackage{csquotes}

\begin{document}

\title [SCDT]{The Signed Cumulative Distribution Transform for 1-D Signal Analysis and Classification}
\date{}
\author[Aldroubi]{Akram Aldroubi}
\address{Akram Aldroubi, Department of Mathematics, Vanderbilt University}
\email{akram.aldroubi@vanderbilt.edu}
\author[D\'iaz Mart\'in]{Roc\'io D\'iaz Mart\'in}
\address{Roc\'io D\'iaz Mart\'in, Universidad Nacional de C\'ordoba, Argentina and Instituto Argentino de Matem\'atica, CONICET, Buenos Aires, Argentina}
\email{rocio.diazmartin@unc.edu.ar}
\author[Medri]{Ivan Medri}
\address{Ivan Medri, Department of Mathematics, Vanderbilt University}
\email{ivan.v.medri@vanderbilt.edu}
\author[Rohde] {Gustavo K. Rohde}
\address{Gustavo K. Rohde, Department of Biomedical Engineering, Department of Electrical and Computer Engineering, University of Virginia}
\email{gr2z@virginia.edu}
\author[Thareja]{Sumati Thareja}
\address{Sumati Thareja, Department of Mathematics, Vanderbilt University}
\email{sumati.thareja@vanderbilt.edu}
	
	\thanks{This work is partially supported by NIH  award GM130825. }


\maketitle

\begin{abstract}  

This paper presents a new mathematical signal transform that is especially suitable for decoding information related to non-rigid signal displacements. We provide a measure theoretic framework to extend the existing Cumulative Distribution Transform \cite{park2018cumulative} to arbitrary (signed) signals on $\RE$. We present both forward (analysis) and inverse (synthesis) formulas for the transform, and describe several of its properties including translation, scaling, convexity, linear separability and others. Finally, we describe a metric in transform space, and demonstrate the application of the transform in classifying (detecting) signals under random displacements.




\end{abstract}

\section{Introduction} 

Mathematical transforms for representing signals and images are useful tools in data-science, engineering, physics, and mathematics.  They  often render certain problems easier to solve in transform space. Fourier transforms \cite{kammler2007first} for example,  render convolution operations into multiplications in Fourier transform space, thereby simplifying the solution of linear shift-invariant systems. They are also well-suited for the detection and analysis of signals that are linear combinations of pure frequencies. Wavelet transforms, on the other hand, are well-suited for detecting and analyzing signal transients at different resolutions. In wavelet space domain, they  provide sparse representations of signals and images for compression and communication \cite{mallat1999wavelet}. Though useful in many areas of mathematics, physics and engineering, most mathematical transformation methods (e.g. Fourier, Wavelet) are linear, and thus often fail to deal with the non-linearities present in modern data science applications related to signal parameter estimation and learning-based data classification. There are several exceptions to this shortcoming. For example, the scattering transform is non-linear and has been successfully applied to machine learning applications \cite{mallat2012group}. 

Inspired by earlier work on transport metrics \cite{wang2013linear},  the Cumulative Distribution Transform (CDT) was introduced for a class of positive, piece-wise continuous, normalized functions \cite{park2018cumulative}. The CDT can be described as follows: let $s_0$ and $s$ denote two 1-dimensional piece-wise continuous functions with domains $\Omega_0$ and  $\Omega$, respectively, such that  $\int_{\Omega_0} s_0(t) dt = \int_{\Omega} s(t)dt = 1$, and such that $s_0,s > 0 $ in their respective domains. We can then relate $s$ and $s_0$ by computing a function $\widehat{s} : \Omega_0 \rightarrow \Omega$ that matches their cumulative integrals:
\begin{equation}
    \int_{\inf(\Omega_0)}^{t} s_0(u) \, du = \int_{\inf(\Omega)}^{\widehat{s}(t)} s(v) \, dv
    \label{eq:cdt_def1}
\end{equation}
The continuous function $\widehat{s}(t)$  is called the Cumulative Distribution Transform (CDT) of $s$ with respect to (some fixed) \textit{reference function} $s_0$. Equation \eqref{eq:cdt_def1} defines a mapping between, positive,  piece-wise continuous, normalized functions and the set of non-decreasing,  one-to-one functions from $\Omega$ and $\Omega_0$. The mapping is invertible on its range, and can be defined in differential form as:
\begin{equation}
    s(t) = (\hat{s}^{-1})^{\prime}(t) \,  s_0(\hat{s}^{-1}(t)), \quad \mbox{ for } \mbox{  }a.e. \mbox{ in } \Omega.
\end{equation}
Like other transforms, it is the properties of the CDT  that make it useful for signal and image data analysis. For example for the CDT, it can be shown that  for translation  $ s(t-\tau) \mapsto  \widehat{s}(y) + \tau$, and for scaling $ as(a t) \mapsto  \widehat{s}(y)/a $. In fact,  an important and unique property of the CDT is that it can represent rigid and non-rigid displacements of the independent variable (in this case $t$) as modifications of the dependent variable in transform space.

In addition, if we consider a convex set $\mathbb{H}$ of invertible mapping functions,  then the set of  signals $\{f^{\prime}(t) \, s (f(t)): \, f^{-1} \in \mathbb{H}\}$ forms a convex set in CDT space \cite{aldroubi2020partitioning}. This property allows one to solve nonlinear, non-convex, signal estimation problems, in a straightforward way, using linear least squares regression \cite{rubaiyat2020parametric}. This property also allows one to solve nonlinear classification problems using linear classifiers in signal transform space \cite{shifat2020radon}.

The CDT  can be related to the optimal transport theory of Monge \cite{park2018cumulative} and has been applied in data analysis, processing and classification problems. Its use is particularly well-suited to mine information present in signals or images when these are produced by physical or biological phenomena related to mass transport. For these cases, popular machine learning methods can be successfully used in transform space for modeling transport modes of variations in signals and images. The CDT and other similar transforms are collectively called  \emph{transport transforms} because of  their connections to Wasserstein distances and optimal transport theory (see below) \cite{wang2013linear,park2018cumulative, kolouri2016radon, kolouri2016continuous, kolouri2016sliced}. They have been used in numerous data science applications, ranging from classification of accelerometer recordings \cite{park2018cumulative}, cancer detection \cite{ozolek2014accurate,tosun2015detection}, drug discovery \cite{basu2014detecting}, knee osteoarthritis prognosis from MRIs \cite{kundu2020enabling}, Brain image analysis \cite{kundu2018discovery}, inverse problems \cite{kolouri2015,huang2020discretized}, optical communications \cite{park2018multiplexing}, particle physics \cite{cai2020linearized}, parametric signal estimation \cite{rubaiyat2020parametric}, and numerous other applications \cite{kolouri2017optimal}.

To illustrate some of the properties of the CDT, consider the task of building a hand signal interpretation system from image data. Sample hand signal images are shown in Figure \ref{fig:intro_fig}, adapted from \cite{park2018cumulative}. The goal is to build a data classification method that can automatically and accurately assign a label (sign) to a given image. Images are pre-processed so as to extract an edge map, and the $X$ and $Y$ projections of the edge maps are computed (middle of first row in Figure \ref{fig:intro_fig}). The  CDT (modified by subtracting the identity function) of the $X,Y$ projections are computed and shown in the right panel of the figure. The test data (both in signal and CDT domain) can be projected onto the most discriminant 2D subspace, computed with the P-LDA technique \cite{wang2011penalized} on training data (middle row). From this it can be seen that the test data become more clearly linearly separable when represented in CDT domain. This visual impression is confirmed by classification results, shown in the bottom table, demonstrating the test data performance using three different linear classification methods. \footnote{The CDT defined in earlier work \cite{park2018cumulative}, and as appears in Figure \ref{fig:intro_fig}, is defined as $\hat{s}(t)-t$, with $\hat{s}$ defined in equation \eqref{eq:cdt_def1}.}

\begin{figure}[!ht]
	\centering
	\includegraphics[width=0.85\linewidth]{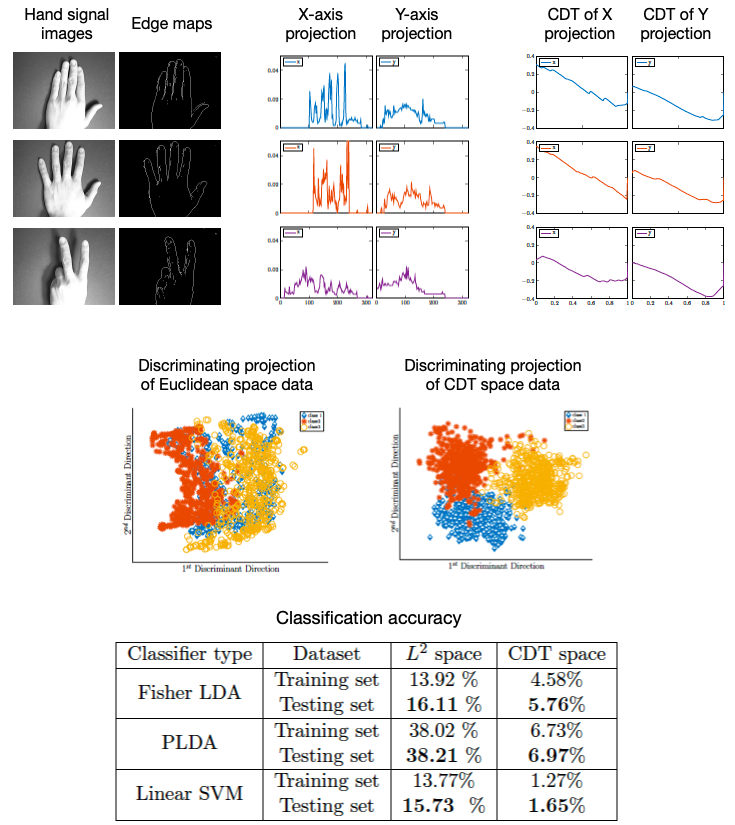}
	\caption{The (modified) Cumulative Distribution Transform (CDT) enables data representation that facilitates learning. Hand signal images are preprocessed for edge map extraction and their respective X and Y projections are computed. The X,Y projections are then transformed using the CDT. Linear classification methods are then applied to the data in CDT space, as well as in original (projection) space for comparison. The middle row displays the 2D linear discriminant embedding \cite{wang2011penalized} of test data in original signal space and transform space. Test (held out from training) data in transform space is clearly more convex and linearly separable than data in original signal space. This is confirmed by test accuracy results of 3 different linear classifiers (bottom row).} 
\label{fig:intro_fig}
\end{figure}

\subsection{Related work}
 The signal transformation described in this paper is related to the optimal transport metric (Wasserstein distance), as will be detailed below in section \ref{sec_distance}. Data analysis techniques that are based on optimal transport have gained popularity in the data science community. Machine learning methods based on optimal transport optimization have been developed in  \cite{arjovsky2017wasserstein, kolouri2019generalized}. Optimal transport methods have also been used  in image processing for image alignment \cite{haker2004optimal}, image  simulation \cite{zhu2007image}, and domain adaptation  \cite{courty2016optimal}.  Extensions of the optimal transport (Wasserstein) metric to unbalanced distributions have been proposed in \cite{thorpe2017transportation,chizat2015unbalanced}. Inspired by \cite{wang2013linear} authors in  \cite{cai2021linearized} proposed a linear optimal transport metric for arbitrary (signed, unbalanced) signals. We also note a generalization of the Wasserstein distance to signed signals, similar to the one described in section \ref{Props&Apps}, has been previously described in \cite{engquist2016optimal, seismic}. 

\subsection{Contributions}
Although useful in many settings, the CDT framework described above has the limitation that the signals themselves must be positive for the entirety of their domain. While not prohibitive in certain settings \cite{rubaiyat2020parametric}, this limitation can hinder the application of the transform to signed functions that can be used to model more general signals and data.

Here we extend the CDT, originally designed for positive probability density functions, to general finite signed measures with no requirements on their total mass. For this reason, we name the new 1D signal transform introduced here as the \emph{signed cumulative distribution transform} (SCDT). We define mathematical formulas for both the forward (analysis) and the inverse (synthesis) transformations. We describe some of the properties of the SCDT, including translation, scaling and composition. Through the use of an extended generative model, we also describe necessary and sufficient conditions whereby signal classes will be convex (and thus separable by a linear classifier). Finally, we  define a distance based on this transform (a version of the distance appears earlier in \cite{engquist2016optimal, seismic}), and demonstrate its application in signal data analysis and classification tasks. Python source code implementing the new transform is available through the PyTransKit package \cite{Codes3}.

{\color{red} \subsection {outline}}

\section{Cumulative Distribution Transforms for Measures} \label{TFM}
In this section we extend the Cumulative Distribution Transform (CDT) to Radon measures on the extended real line $\RE$. Because we use concepts related to the theory of transport, we start with an extension of CDT to probability measures. We then extend the CDT to include non-negative finite measures, and finally signed measures.  

In the cumulative distribution transform for measures, a reference measure $\mu_0$ is fixed, and the transform of a measure $\nu$ relative to this fixed reference measure $\mu_0$ will be a non-decreasing function on $\RE$ denoted by $\widehat{\nu}$. 

\subsection{The cumulative distribution transform on probability measures}
Given any probability measure $\eta$ on $\RE$, its Cumulative Distribution Function (CDF) is the function $F_\eta:\RE\to \RE$ given by 
\begin{equation} \label{CD}
    F_\eta(x)=\eta([-\infty,x]).
\end{equation}
The function $F_\eta$ is non-decreasing and  right continuous on $\RE$, and may not be invertible. However, a generalized inverse for such a function (in fact for any function) can be defined (see e.g., \cite{Embrechts:13, thorpeNotes}). 
\begin{definition} \label {MGI}
 For a function $F:\RE \to \RE$, the monotone generalized inverse of $F$ is the function $F^{\inv}:\RE\to\RE$  defined as
    \begin{equation*}
        F^{\inv}(y)  = \inf\{x \in \RE: \, F(x)>y\} .
    \end{equation*}
    In particular (since $\inf \emptyset = \infty$),  $F^{\inv}(\infty) = \infty.$ 
\end{definition}

The monotone generalized inverse $F^\inv$ has some remarkable properties. In particular, for any function $F$, $F^\inv$ is a non-decreasing  function on $\RE$, and  if   $F$ is continuous and strictly increasing,  then $F^\inv$ and its standard inverse coincide (see Appendix \ref{appendix}). Since the monotone generalized inverses are functions on $\RE$, our measures will be defined on $\RE$ instead of $\R$. Those measures are characterized by their CDFs (which are functions on $\RE$) according to Proposition \ref{characterization_meausres_on_RE}. 

Given a reference measure $\mu_0$,  the cumulative distribution transform $\widehat{\nu}$ of a measure $\nu$ with respect to $\mu_0$ is defined as
\begin{equation} \label {vhat}
		\widehat{\nu}=(F_\nu)^\inv \circ F_{\mu_0},
\end{equation}
	where $F_{\mu_0}$ and $F_\nu$ are the CDFs of the respective measures. Under the assumption that $\mu_0$ does not give mass to atoms, the function $\eqref{vhat}$ is the solution of the 1D optimal transport problem of Monge (see for example \cite{thorpeNotes, villani2003topics}).   In fact, the measure $\nu$   can be recovered by the push-forward $\widehat\nu_\#\mu_0$ of the measure $\mu_0$ by the function $\widehat \nu$ in \eqref {vhat}  (see Theorem \ref{Prob_bijection}), i.e., for any Borel measurable set $E\subset \RE$
	
	\begin{equation} \label {PF}
	\nu(E)=\widehat\nu_\#\mu_0(E):=\mu_0\big(\widehat \nu^{-1}(E)\big).
    \end{equation}
    
If the reference measure $\mu_0$ and the target measure $\nu$ are continuous with respect to the Lebesgue measure with densities $s_0$ and $s$, respectively, then equation \eqref{eq:cdt_def1} can be rewritten as follows
    $$F_{\mu_0}=F_\nu(\widehat{\nu})$$
Therefore, equation \eqref{vhat} extends the definition  of the CDT  for functions as described by \eqref{eq:cdt_def1} to the case of  probability measures on $\RE$.

Writing \eqref{vhat} in operator notation as $\widehat{\nu}=\TMP_{\mu_0}(\nu)$, we define the CDT operator by
\begin{equation} \label{TOPM}
\TMP_{\mu_0} \; :\mathcal{P}(\RE) \to \ND 
\end{equation}
where $\mathcal{P}(\RE)$ is the set of probability measures on $\RE$ and 
\begin{equation}\label{TOPMR}
\ND:=\{f:\RE\to\RE \text{ non-decreasing } \mu_0 \text{-}a.e \}
\end{equation}
is the set of non-decreasing functions $a.e. \, $  with respect to $\mu_0$.
The next theorem states that the operator $\TMP_{\mu_0}$ above is a bijection and hence can be viewed as a transform operator.
\begin{theorem}\label{Prob_bijection}
    If $\mu_0$ is a probability measure on $\RE$ that does not give mass to atoms, then the operator  $\TMP_{\mu_0} :\mathcal{P}(\RE) \to \ND$ in \eqref{TOPM} is a bijection. 
    Moreover, if $\varphi \in \ND$, then $\TMP_{\mu_0}^{-1}(\varphi)=\varphi_\#\mu_0$ (see \eqref {PF}). 
\end{theorem}

A measure $\eta$ defined on a Borel measurable subset   $\Omega\subseteq\RE$ can be considered as a measure on $\RE$ by extending it by zero. Equivalently, this extension can be written as the  push-forward $\iota_{\#} \eta$ of $\eta$  by the inclusion map $\iota:\Omega\to \RE$ ($\iota_{\#} \eta\in\mathcal{P}(\RE)$). Thus, using this extension, we consider the set of probability measures $\mathcal{P}(\Omega)$ defined on $\Omega$ as a subset of $\mathcal{P}(\RE)$, i.e., $\mathcal{P}(\Omega)\subseteq \mathcal{P}(\RE)$. Analogously, we will say that any measure $\eta\in \mathcal{P}(\RE)$ satisfying $\eta(\Omega)=1$ and $\eta(\Omega^c)=0$ belongs to $\mathcal{P}(\Omega)$ by considering its restriction. From these considerations, we obtain the following corollary of Theorem \ref {Prob_bijection}.

\begin{corollary} \label{prob_restrict} Let $\Omega_0,\Omega$ be two Borel sets in $\RE$ and $\mu_0$ be a probability measure on  $\Omega_0$ that does not give mass to atoms. Then the restriction $\TMP_{\mu_0}|_{\mathcal{P}(\Omega)}$ of the transform operator $\TMP_{\mu_0}$ to $\mathcal{P}(\Omega)$  is a bijection from $\mathcal{P}(\Omega)$ to the set $\ND(\Omega):=\{f \in \ND : \,  f(x) \in \Omega \ \mu_0 \text{-} a.e. \text{ for } x\in \RE \}$ (see \eqref{TOPM} and \eqref{TOPMR}).
\end{corollary}
\begin{remark} \label{NO0O}
In the corollary above, since $f\in \ND(\Omega)$ is defined $\mu_0$-$a.e.$ and $\mu_0(\Omega_0^c) = 0$,  we can consider  the image of the transform operator to be the set 
$$\ND(\Omega_0,\Omega):=\{f:\Omega_0\to \Omega \text{ non-decreasing } \mu_0\text{-}a.e. \},$$
by considering the restriction $\TMP_{\mu_0}(\nu)_{|_{\Omega_0}} = (F_\nu ^\inv \circ F_{\mu_0})_{|_{\Omega_0}}$ for every $\nu \in \mathcal P(\Omega)$.
\end{remark}

\begin{example}
This extension of the CDT to a transform for probability measures allows to consider the transform of a delta-function. Let $\mu_0$ be the Lebesgue measure on $[0,1]$ (that is, consider the uniform distribution). Let $\nu=\delta_0$ be the Dirac measure concentrated on $\{0\}$. 
On the one hand, $F_{\mu_0}$ is zero on $[-\infty,0]$, the identity function on $(0,1)$ and $1$ on $[1,\infty]$, 
and 
$F_{\delta_0}$ is the Heaviside step function $H$ (that is, the characteristic function of $[0,\infty]$).
Therefore,
$$\widehat{\delta_0}(x)=
\begin{cases}
H^\inv(0) & \text{ if } x\leq 0\\
H^\inv(x) & \text{ if } x\in (0,1)\\
H^\inv(1) & \text{ if } x\geq 1
\end{cases}$$
where
$$H^\inv(0)=\inf\{x: H(x)>0\}=-\infty, \qquad H^\inv(1)=\inf\{x: H(x)>1\}=\infty$$
and
$$H^\inv(y)=\inf\{x: H(x)>y\}=0 \qquad \forall y\in (0,1)$$
Thus,
\begin{equation}\label{delta_hat}
    \widehat{\delta_0}(x)=
\begin{cases}
-\infty & \text{ if } x\leq 0\\
0 & \text{ if } x\in (0,1)\\
\infty & \text{ if } x\geq 1
\end{cases}
\end{equation}
In particular, $\widehat{\delta_0}(x)=0$, a.e.-$\mu_0$.
On the other hand, the push-forward of $\mu_0$ by the zero function 
$$0_\#\mu_0(E)=\mu_0\left(\{x:\, 0\in E\}\right)=\begin{cases}
\mu_0(\RE) & \text{ if } 0\in E\\
\mu(\emptyset) & \text{ if } 0\not\in E 
\end{cases}=\begin{cases}
1 & \text{ if } 0\in E\\
0 & \text{ if } 0\notin E
\end{cases}=\delta_0(E)$$
for every  Borel set $E\subset \RE$. Thus, $\delta_0=0_\#\mu_0$.
\end{example}

\subsection{Transform for positive finite measures with arbitrary total mass}

The CDT for probability measures can be extended naturally to the case where the reference and the target  are finite and positive Borel measures $\BM(\RE)$. Specifically, let $\mu_0,\nu$ be two finite positive measures such that $\mu_0$ is non-trivial, that is
$$0<\|\mu_0\|=\mu_0(\RE)<\infty \qquad \text{ and } \qquad  \|\nu\|=\nu(\RE)<\infty$$
If $M_\alpha$ denotes the scaling function by the factor $\alpha$ (i.e., $M_\alpha(x)=\alpha x $), then the transform of $\nu$ with respect to the reference $\mu_0$ is defined to be 
 \begin{equation} \label{PMT}
     \widehat{\nu} = \begin{cases}
                    \left(\nu^\star, \|\nu\|\right), \quad  &\nu\ne 0,\\
                                    \\
                     \left(0, 0\right), \quad &\nu= 0,
          \end{cases}
\end{equation}
and where
  \[
  \nu^\star=F_{\nu}^\inv\circ M_{\|\nu\|}\circ M_{\frac 1 {\|\mu_0\|}}\circ F_{\mu_0}.
  \]

 When $\nu\ne 0$, the non-decreasing function $\nu^*$ is simply the CDT of the probability measure $\frac {1} {\|\nu\|}\nu$ with respect to the reference $\frac {1} {\|\mu_0\|}\mu_0$ (see \eqref {vhat}), while the number $\|\nu\|$ is to keep track of the total mass of $\nu$. The next theorem shows that \eqref {PMT} gives rise to a bijection and hence can be thought of as a transform.  We abuse language and still call this transform the CDT since it will be clear from the context which transform is being used. Abusing notation by using  $\TMP_{\mu_0}$ again to denote the operator defined by $ \TMP_{\mu_0}(\nu)=\widehat \nu$ we get the following bijection between the space of  finite Borel measures $\BM(\RE)$ and the set 
 \begin{equation}
 \label {TPMAMR}
 \STPM:=\Big(\ND\times \R^+\Big)\cup \{(0,0)\},
 \end{equation}
where $\ND$ is as in \eqref{TOPMR},  and $\R^+=(0,\infty)$.

\begin{theorem}\label{bij_non-nomnalized}
If $\mu_0$ is a non-trivial finite positive measure on $\RE$ that does not give mass to atoms, then the  operator $\TMP_{\mu_0}$ is a bijection from $\mathcal{M}(\RE)$ to $\STPM$, whose inverse  is given by 
    \begin{equation}\label{inv_positive}
    (f,r)\mapsto rf_\# \Big(\frac{\mu_0}{\|\mu_0\|}\Big)\qquad \text {for } (f,r)\ne (0,0),        
    \end{equation}
    and  the inverse of $(0,0)$ is the zero measure.
\end{theorem}

Defining $\STPM (\Omega):=\Big(\ND(\Omega)\times \R^+\Big)\cup \{(0,0)\}$, we get the following corollary.

\begin{corollary} \label{FMsub}
Let $\Omega_0,\Omega$ be two Borel sets in $\RE$ and $\mu_0$ be a non-trivial finite positive Borel  measure on $\Omega_0$ that does not give mass to atoms. Then the CDT with respect to $\mu_0$ is a bijection from $\mathcal{M}(\Omega)$ (the set of positive finite measures on $\Omega$) to  $\STPM (\Omega)$.
\end{corollary}

As in Remark \ref {NO0O},  instead of $\mathcal{N}_{\mu_0}(\Omega)$ in the first component of the image of the transform, we could consider directly the set $\mathcal{N}_{\mu_0}(\Omega_0,\Omega)$ 
by considering restrictions: $(\nu^\star)_{|_{\Omega_0}}$ of $\nu^\star$ in \eqref{PMT}.

\subsection{Transform for signed measures} 
To define the transform on signed measures, we use  the Jordan decomposition of a signed measure \cite{royden:10}. Specifically,  
    let $\nu$ be a signed, finite Borel measure on $\RE$, with Jordan decomposition $\nu=\nu^+ - \nu^-$, then the image $\TMP_{\mu_0}(\nu)$ of  $\nu$ with respect to a non-trivial positive measure $\mu_0$ is defined as
    \begin{equation}\label{one reference}
        \TMP_{\mu_0}(\nu) = \left(\TMP_{\mu_0}(\nu^+), \TMP_{\mu_0}(\nu^-)\right)
    \end{equation}
where $\TMP_{\mu_0}(\nu^\pm)$ are the CDTs of positive finite measures defined in \eqref {PMT}. In particular,  $\TMP_{\mu_0}(\nu) \in \STPM^2$ (see \eqref{TPMAMR}). Denoting the set of finite signed measures on $\RE$ by $\SM(\RE)$, and defining
 \begin{equation} \label{IMU}
        \mathcal{I}_{\mu_0}:=\{(f,r,g,s)\in \STPM^2:  \,     f_\#\mu_0 \perp g_\#\mu_0 \, \text { for } r,s\in\R^+\}
    \end{equation}
where $f_\#\mu_0 \perp g_\#\mu_0$ denotes that $f_\#\mu_0$ and $g_\#\mu_0$ are mutually singular, we have the following theorem. 

\begin{theorem}\label{bijection_signed_measure}
 If $\mu_0$ is a non-trivial finite positive measure on $\RE$ that does not give mass to atoms, then the operator  $\TMP_{\mu_0}:\mathcal{SM}(\RE) \to \mathcal{I}_{\mu_0}$ given in \eqref{one reference} is a bijection. Hence, $\TMP_{\mu_0}$ is a transform.
 
    Moreover, the inverse transform is given by $$(f,r,g,s)\mapsto r\, f_\# \Big(\frac{\mu_0}{\|\mu_0\|}\Big) - s \, g_\#\Big( \frac{\mu_0}{\|\mu_0\|}\Big) \quad \text { for } r,s\in \R^+,$$
and the inverses of $(f,r,0,0)$, $(0,0,g,s)$, and $(0,0,0,0)$ are the measure $f_\# \Big(\frac{\mu_0}{\|\mu_0\|}\Big)$, the measure $- s \, g_\#\Big( \frac{\mu_0}{\|\mu_0\|}\Big)$, and the zero measure, respectively.
\end{theorem}

We call the operator $\TMP_{\mu_0}$ given in \eqref{one reference} the Signed Cumulative Distribution Transform (SCDT). 

As in previous subsections, restricting the measures to Borel sets $\Omega_0$  and $\Omega$ for $\mu_0$ and $\nu$ respectively, and defining 

\begin {equation} \label {IMUOmega}
\mathcal{I}_{\mu_0}(\Omega):=\{(f,r,g,s)\in \STPM^2(\Omega):  \,     f_\#\mu_0 \perp g_\#\mu_0 \, \text { for } r,s\in\R^+\},
\end{equation}
we obtain the following  bijection result.

\begin{corollary}\label{corollary_signed_measure}
Let $\Omega_0,\Omega$ be two Borel sets in $\RE$ and $\mu_0$ be a non-trivial finite positive Borel  measure on $\Omega_0$ that does not give mass to atoms. Then the restriction of the transform to $\mathcal{SM}(\Omega)$ (the set of signed finite measures on $\Omega$) 
is a bijection onto $\mathcal{I}_{\mu_0}(\Omega)$.
\end{corollary}

\section{Properties and Applications}  \label{Props&Apps}

There are several properties of the CDT for positive PDFs that makes it a useful tool. The new transforms derived in this paper  retain some of these useful properties. In particular,  the computational example below, shows how these transforms can be useful in data classification. 
\subsection{Properties of the SCDT} \label {POF}
In this section we list two of these properties: 1) the  composition of  property and 2) the convexification property.

The composition property relates the transform of a measure $\eta$ to that of a measure $
\nu$ when their cumulations are related by a composition of functions. This property is useful for applications in which a set of signals is generated by a signal template that is modified by a transport-like phenomenon. Translations and scalings are examples of such  classes \cite{park2018cumulative} (see Figures \ref{translation} and \ref{dilation}). 

\begin{proposition} (Composition property) \label{Composition_Signed}
    Consider a finite positive reference measure $\mu_0 $ which does not give mass to atoms. Let $\nu$ be a signed measure.  Assume that  $g : \RE \rightarrow \RE $ is a strictly increasing surjection,  and  $\eta$  a signed measure such that $ F_{\eta}(x)=F_{\nu}(g(x))$. Then  $\|\eta^\pm\|=\|\nu^\pm\|$ and  the SCDT of $\eta$ with respect to $\mu_0$ is  given by 
    $$\TMP_{\mu_0}({\eta}) = (g^\inv \circ(\nu^+)^\star,\|\nu^+\|,g^\inv \circ (\nu^-)^\star,\|\nu^-\|),$$
     when $\|\nu^+\|\ne 0$ and $\|\nu^-\|\ne 0.$ If $\TMP_{\mu_0}( {\nu})=((\nu^+)^*, \|\nu^+\|, 0, 0)$, $\TMP_{\mu_0}( {\nu})=(0,0,(\nu^-)^\star, \|\nu^-\|)$, or $\TMP_{\mu_0}( {\nu})=(0,0,0,0)$ then $\TMP_{\mu_0}({\eta})$ is given by $(g^\inv \circ (\nu^+)^\star,\|\nu^+\|,0,0)$, $(0,0,g^\inv \circ (\nu^-)^\star,\|\nu^-\|),$ $(0,0,0,0)$, respectively.
\end{proposition}

\begin{corollary} (Translation)
\label{coro_translation}
If $\mu_0$ and $\nu$ are as above, and $g : \RE \rightarrow \RE $ is a translation by $a\in\R$, i.e. $g(x) = x-a$, then for $\eta \in \mathcal{SM}(\RE)$ such that $ F_{\eta}(x)=F_{\nu}(x-a)$, the SCDT of $\eta$ with respect to $\mu_0$ is  given by  
$\TMP_{\mu_0}({\eta}) = ((\nu^+)^\star + a,\|\nu^+\|,(\nu^-)^\star + a,\|\nu^-\|)$. (See figure \ref{translation}.)
\end{corollary}

\begin{corollary}\label{coro_dilation} (Dilation)
If $\mu_0$ and $\nu$ are as above, $g : \RE \rightarrow \RE $ is a dilation by $a \in (0,\infty)$, i.e. $g(x) = \frac{x}{a},$ then for $\eta \in \mathcal{SM}(\RE)$ such that $ F_{\eta}(x)=F_{\nu}(\frac{x}{a})$, the SCDT of $\eta$ with respect to $\mu_0$ is  given by  
$\TMP_{\mu_0}({\eta}) = (a(\nu^+)^\star,\|\nu^+\|,a(\nu^-)^\star,\|\nu^-\|).$ (See figure \ref{dilation}.)
\end{corollary}

\begin{figure}[H]
	\centering
	\includegraphics[width=1.0\linewidth]{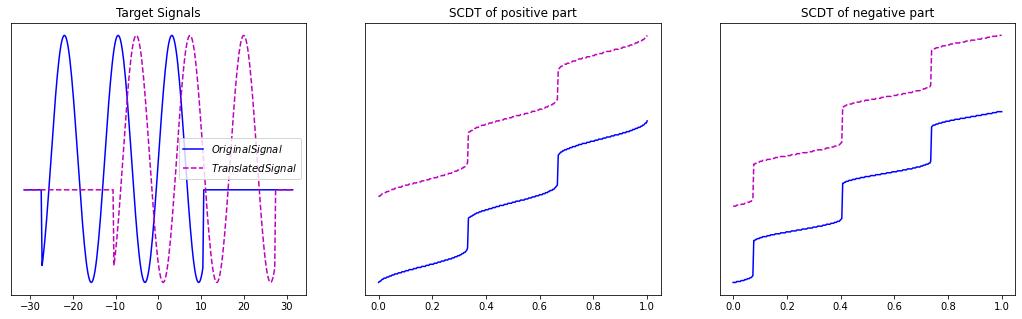}
	\caption{Translation property of the SCDT. Left panel: A signal and its translation. Right panel: A transforms of the signals and its translation (the reference measure $\mu_0$ was taken with uniform distribution on $[0,1]$).}
\label{translation}
\end{figure}

\begin{figure}[H]
	\centering
	\includegraphics[width=1.0\linewidth]{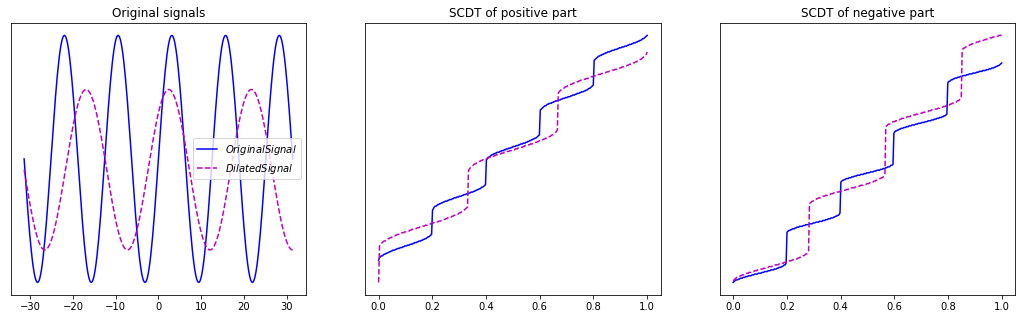}
	\caption{Scaling property of the SCDT. Left panel: A signal and its scaling. Right panel: A transforms of the signals and its scaling (the reference measure $\mu_0$ was taken with uniform distribution on $[0,1]$).}
\label{dilation}
\end{figure}

It is also worth noting that in certain applications data sets can be rendered convex in the transform domain. For example, sets generated by translations of a template signal can have a complex geometry in the signal domain. However, it has a very simple convex structure in the transform domain as depicted in Figure \ref{fig:convexity_fig}.  This convexification property is useful in classification problems since two disjoint convex data sets can be  separated by a linear classifier.  The following convexity property is a generalization of the convexity property that was proved in \cite{aldroubi2020partitioning} for signals that consist of normalized, compactly supported, non-negative Lebesgue measurable functions.

\begin{proposition}\label{convexity_Signed}
(Convexity property)
Let $\mu_0$ be a finite positive reference measure  that does not give mass to atoms and  let $\nu$ be a signed measure. Let $\mathbb H$ be a subset of strictly increasing surjections  on $\RE$, and
\begin{equation}\label{S_nu,H}
   \mathcal S_{\nu,\mathbb{ H}}:=\{\eta\in\mathcal{SM}(\RE): F_\eta = F_\nu \circ h, \, h\in \mathbb{H} \}. 
\end{equation}
Then, $\widehat{\mathcal S_{\nu,\mathbb{ H}}}:=\{\widehat{\eta}: \, \eta\in \mathcal S_{\nu,\mathbb{H}}\}$
is convex for every $\nu$ if and only if $\mathbb{H}^{-1}:=\{h^{-1}: \, h\in \mathbb{H}\}$ is convex.
\end{proposition}

We remark that the set $\mathcal S_{\nu,\mathbb{ H}}$ can be interpreted as an algebraic generative model for signal data. Here $\nu$ specifies a measure (signal) that can be considered as a template for its class, which is denoted by $\mathcal S_{\nu,\mathbb{ H}}$. The elements of  $\mathcal S_{\nu,\mathbb{ H}}$ are formed  by the action of functions in  $\mathbb H$, as in Proposition \ref {convexity_Signed}. For example, $\mathbb H$ can be the set of all possible translations, or positive scalings. If  $\mathbb H$  is such that $\mathbb{H}^{-1}$ is convex, then the set $\widehat{\mathcal S _{\nu,\mathbb{ H}}}$ is also convex. As explained in \cite{aldroubi2020partitioning}, $\mathbb{H}^{-1}$ is convex if $\widehat{\mathcal S_{\nu,\mathbb{ H}}}$ is a convex group (numerous examples are described in \cite{aldroubi2020partitioning}).

\begin{figure}[H]
	\centering
	\includegraphics[width=0.6\linewidth]{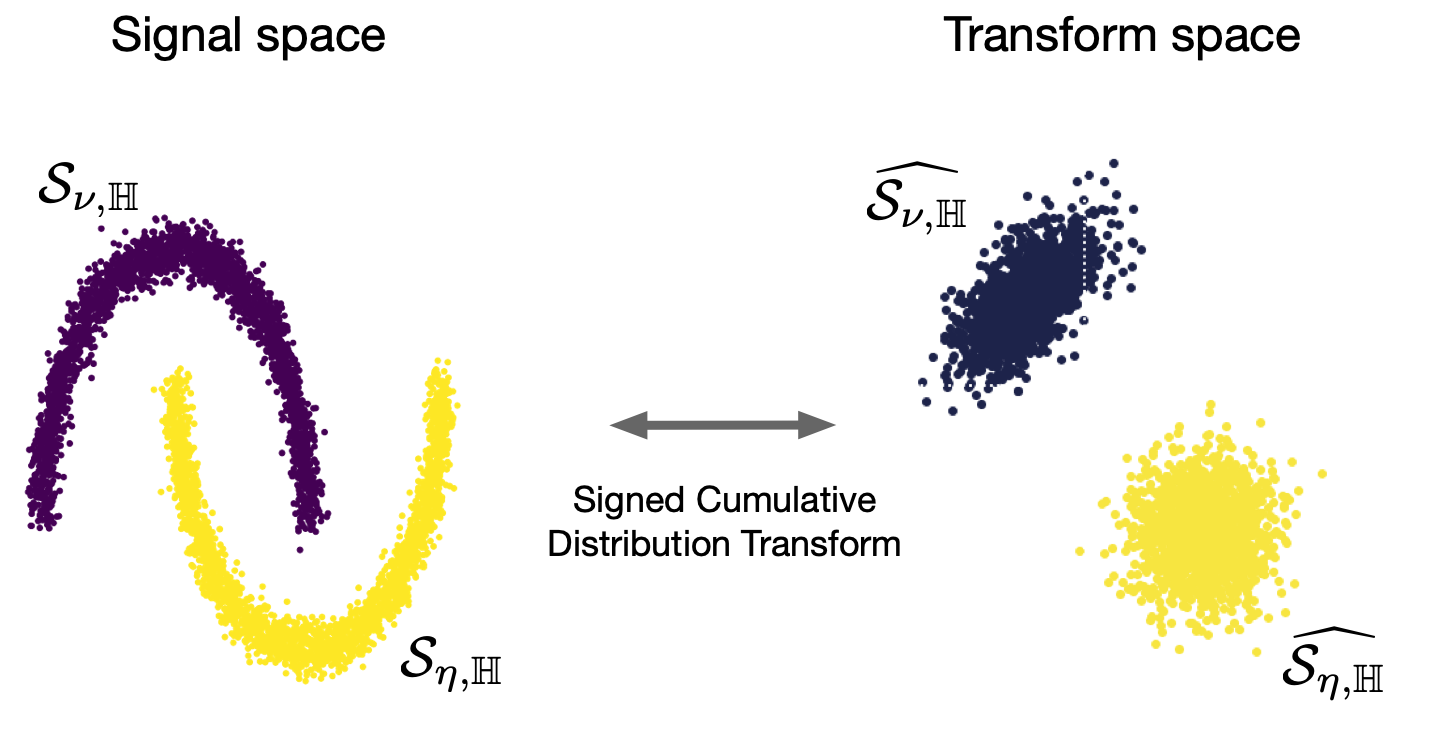}
	\caption{ The set of signals generated from the algebraic generative model stated in Proposition \ref{convexity_Signed} becomes convex  in the SCDT space.} 
\label{fig:convexity_fig}
\end{figure}

\subsection{Metric}\label{sec_distance}
The Wasserstein distance  in the space $\mathcal{P}(\Omega)$ is intimately related to the $L^2$ distance in the transport transform domain \cite{park2018cumulative, thorpeNotes, villani2003topics}. This relation is useful in some applications since it render certain optimization problems involving the Wasserstein distance into  standard least squares optimization problems. In this section, we will recall the definition of the Wasserstein distance for probability measures and its relation to the $L^2$ distance in the transform domain. We will extend this result using the generalized transport transforms of the previous section.

\begin{definition}[Wasserstein distance]
Let $\Omega\subseteq \R$ be a Borel set endowed with the Euclidean distance. 
Let $\mathcal{P}_2(\Omega)$ denote the subset of all probability measure on $\Omega$ with finite second moment
$$\mathcal{P}_2(\Omega):=\left\{\nu\in\mathcal{P}(\Omega): \, \int_{\Omega} |x|^2 \, d\nu(x)<\infty\right\}$$
The 
Wasserstein distance $d_{W^2}$ is defined as
\begin{equation}\label{was_original}
 d_{W^2}(\nu,\eta):=\left(\inf_{\pi\in\Pi(\nu,\eta)}\int_{\R\times\R} |x-y|^2 \, d\pi(x,y)\right)^{1/2} \qquad \forall \nu,\eta\in \mathcal{P}_2(\Omega), 
\end{equation}
where $\Pi(\nu,\eta)$ denotes the collection of all measures on $\Omega\times\Omega$ with marginals $\nu$   and $\eta$  on the first and second factors respectively.
\end{definition}

The following embedding is well-known \cite{thorpeNotes, villani2003topics}.

\begin{proposition}\label{prop_wass}
Let $\Omega\subseteq \R$ be a Borel set, and $\mu_0\in \mathcal{P}_2(\Omega)$ a reference measure that does not give mass to atoms. If $\nu, \eta\in\mathcal{P}_2(\Omega)$, then 
\begin{equation}\label{wass_prob}
    d_{W^2}(\nu,\eta) =\|\TMP_{\mu_0}(\nu) -\TMP_{\mu_0}(\eta) \|_{L^2(\mu_0)}
\end{equation}
In particular, the transport transform $\TMP_{\mu_0}$ is an isometry between $\mathcal{P}_2(\Omega)$ with the Wasserstein metric and the image of the transport transform $\mathcal{N}_{\mu_0}(\Omega)$ endowed with the $L^2(\mu_0)$ metric.
\end{proposition}

\subsection{Metric for $\BM_2(\Omega)$ and $\SM_2(\Omega)$} 
\label{sec:metric}
Let $\BM_2(\Omega)$ and $\SM_2(\Omega)$ be the subsets of $\BM(\Omega)$ and $\SM(\Omega)$, respectively, of measures having finite second moment. 
For $\mu_0\in \mathcal{P}_2(\Omega)$ that does not give mass to atoms, consider the Cartesian product  $L^2(\mu_0)\times\R$  endowed with the norm
\begin{equation*}
    \|(f,r)\|_{L^2(\mu_0)\times\R}=\sqrt{\|f\|_{L^2(\mu_0)}^2+r^2 \, }  .
\end{equation*}
For $\nu,\eta\in\BM_2(\Omega)$  we define the following distance function:
\begin{equation}\label{wass_prob1}
    D_{W^2}(\nu,\eta) :=\|\TMP_{\mu_0}(\nu) -\TMP_{\mu_0}(\eta) \|_{L^2(\mu_0)\times\R}
\end{equation}
In particular,  from Proposition \ref{prop_wass} it follows that for non-trivial measures $\nu,\eta\in \BM_2(\Omega)$, \eqref{wass_prob1} coincides with 
\begin{equation}\label{wass_definition2}
   D_{W^2}(\nu,\eta) =\sqrt{ \left(d_{W^2}\left(\frac{\nu}{\|\nu\|},\frac{\eta}{\|\eta\|}\right)\right)^2 + \left|\,  \|\nu\| - \|\eta\|\, \right|^2 \,  }
\end{equation}
where $d_{W^2}$ is the Wasserstein distance defined in Equation \eqref {was_original}. 
This identity and the definition of the transform when $\nu$ or $\eta$ are the zero measures, imply that $D_{W^2}$ does not depend on the choice of the reference $\mu_0$.
 In particular, if $\nu$ is non-trivial and $\eta$ is the zero measure
$$
    \left(D_{W^2}(\nu,0)\right)^2=\|\TMP_{\mu_0}(\nu)  \|_{L^2(\mu_0)}^2 =\int_\R |F_{\frac{\nu}{\|\nu\|}}^\inv\circ F_{\mu_0}|^2 \, d\mu_0 +\|\nu\|^2\\
    =\int_0^1 |F_{\frac{\nu}{\|\nu\|}}^\inv(t)|^2 dt+\|\nu\|^2,
$$
which holds by applying Lemma \ref{Lema_aux1} in the same way it is done in the proof of Proposition \ref{prop_wass}. 
Using the Definition \eqref {wass_prob1},  $\TMP_{\mu_0}$ becomes  an isometry from $(\BM_2(\Omega), D_{W^2}) $ to  $(\mathcal {T}_{\mu_0}(\Omega),\|\cdot\|_{L^2(\mu_0)\times\R)})$. \\

Analogously, by considering the space $(L^2(\mu_0)\times\R)^2$ endowed with
\begin{equation*}
    \|(f,r,g,s)\|_{(L^2(\mu_0)\times\R)^2}=\sqrt{\|f\|_{L^2(\mu_0)}^2+r^2+\|g\|_{L^2(\mu_0)}^2+s^2  \, },  
\end{equation*}
we define 
\begin{equation}\label{wass_prob2}
    D_{S}(\nu,\eta) :=\|\TMP_{\mu_0}(\nu) -\TMP_{\mu_0}(\eta) \|_{(L^2(\mu_0)\times\R)^2} \qquad \forall \nu,\eta\in\SM_2(\Omega)
\end{equation}
which endows $\SM_2(\Omega)$ with a distance on  and which is independent from the choice of the reference $\mu_0$. Indeed, using the Hahn-Jordan decomposition for  $\nu, \eta \in \SM_2(\Omega)$, \eqref{wass_prob2} is exactly 
\begin{equation}\label{wass_definition3}
  D_{S}(\nu,\eta) = \sqrt{\left(D_{W^2}(\nu^+,\eta^+)\right)^2 + \left(D_{W^2}(\nu^-,\eta^-)\right)^2 \, }.
\end{equation}


\subsection{Applications}

The CDT and SCDT have many applications in signal and image analysis which can be broadly categorized into signal estimation and detection (classification) problems. With regards to signal estimation, in \cite{rubaiyat2020parametric} the authors applied the CDT to estimating parameters (time delay, frequency, chirp) pertaining to a measured signal. In that work, the CDT was used to `linearize' the problem so that a global optimal estimate for signal parameters can be estimated in CDT space using a simple linear least squares technique. Here we show how the SCDT can be utilized to similarly facilitate the machine learning of classifiers by `linearizing' the problem in transform space, as illustrated in Figure \ref{fig:convexity_fig}. To that end, we utilize the following property of the SCDT.

Let $\mathbb H,$ $\mathcal S_{\nu,\mathbb{ H}}$ and $ \mathcal S_{\eta,\mathbb{ H}}$ be as defined in Proposition \ref{convexity_Signed}. The sets $\mathcal S_{\nu,\mathbb{ H}}$ and $ \mathcal S_{\eta,\mathbb{ H}}$ can be interpreted as algebraic generative models for signal data. The theorem stated below becomes an easy consequence of Proposition \ref {convexity_Signed} and the Hahn Banach separation theorem.

\begin{theorem} \label{linearSeparability}
Let $\mu_0 \in \mathcal P_2(\R)$ be a reference measure that does not give mass to atoms,  $\nu,\eta \in \SM_2(\Omega)$ be two signed measures, and $\mathbb H^{-1} $ be a convex set of strictly increasing bijections on $\R$.  If $\mathcal D_\nu, \mathcal D_\eta$ are two non-empty, finite sets drawn from  two disjoint generative models $\mathcal S_{\nu,\mathbb{ H}}$ and $ \mathcal S_{\eta,\mathbb{ H}}$,  respectively, then the corresponding sets $\widehat {\mathcal D_\nu}, \widehat{\mathcal D_\eta}$ in SCDT space are linearly separable.
\end{theorem}

The theorem above states that so long as data is generated according to the algebraic generative model defined in Proposition \ref{convexity_Signed}, a training procedure using a linear classifier, using data in SCDT space, is a well-posed problem in the sense that there is a guarantee that a solution exists, although it may not be unique. The theorem does not state how to compute such linear function, rather it states that there will exist a linear classifier that will separate $\widehat {\mathcal D_\nu}, \widehat{\mathcal D_\eta}$ so long as $\mathcal S_{\nu,\mathbb{ H}}$ and $ \mathcal S_{\eta,\mathbb{ H}}$ are disjoint.

To demonstrate the ability of the SCDT to render signal classes linearly separable we consider the problem of distinguishing signals of the kind demonstrated in Figure \ref{fig:signals_for_classification}. Let signals $\sigma_1,\sigma_2,\sigma_3$  be associated with the measures $d\nu_1 = \sigma_1(t)dt,d\nu_2 = \sigma_2(t)dt,d\nu_3 = \sigma_3(t)dt$, and let $\mathcal S_{\nu_1,\mathbb{ H}}$, $\mathcal S_{\nu_2,\mathbb{ H}}$, and  $\mathcal S_{\nu_3,\mathbb{ H}}$ represent corresponding signal classes  generated using the set of diffeomorphisms $\mathbb{H}=\{h(t) = at + b: a,b \in \mathbb{R}, a > 0 \}.$ In short, three prototype signals are defined as a Gabor wave, a Sawtooth wave, and a Square wave, all  multiplied by a Gaussian window function, respectively. These prototypes are randomly translated and scaled. For the computer simulations shown below, $t \in [-0.5,5]$, and $a$ and $b$ are uniformly distributed in $[0.75,2]$ and $[-0.25, 0.25]$ respectively. A total of $N=500$ sample signals are generated with $250$ used for training and $250$ for testing. Randomly distributed Gaussian noise, with standard deviation of $0.02$ was added to each signal. The Fisher Linear Discriminant Analysis \cite{fisher1936} computed using the sklearn python package \cite{Codes2}, shows that classification accuracy on the test set using the data in original signal space is $32\%$, while the test set accuracy of the same classification algorithm applied to signals in SCDT space is $99\%$. The projections of the test set for both native signal space and SCDT space are shown in Figure \ref{fig:LDA_classification} below. From these figures, and from the test set classification accuracy, we can see the SCDT significantly enhances the ability of a linear classifier to operate correctly.

\begin{figure}[H]
	\centering
	\includegraphics[width=1.0\linewidth]{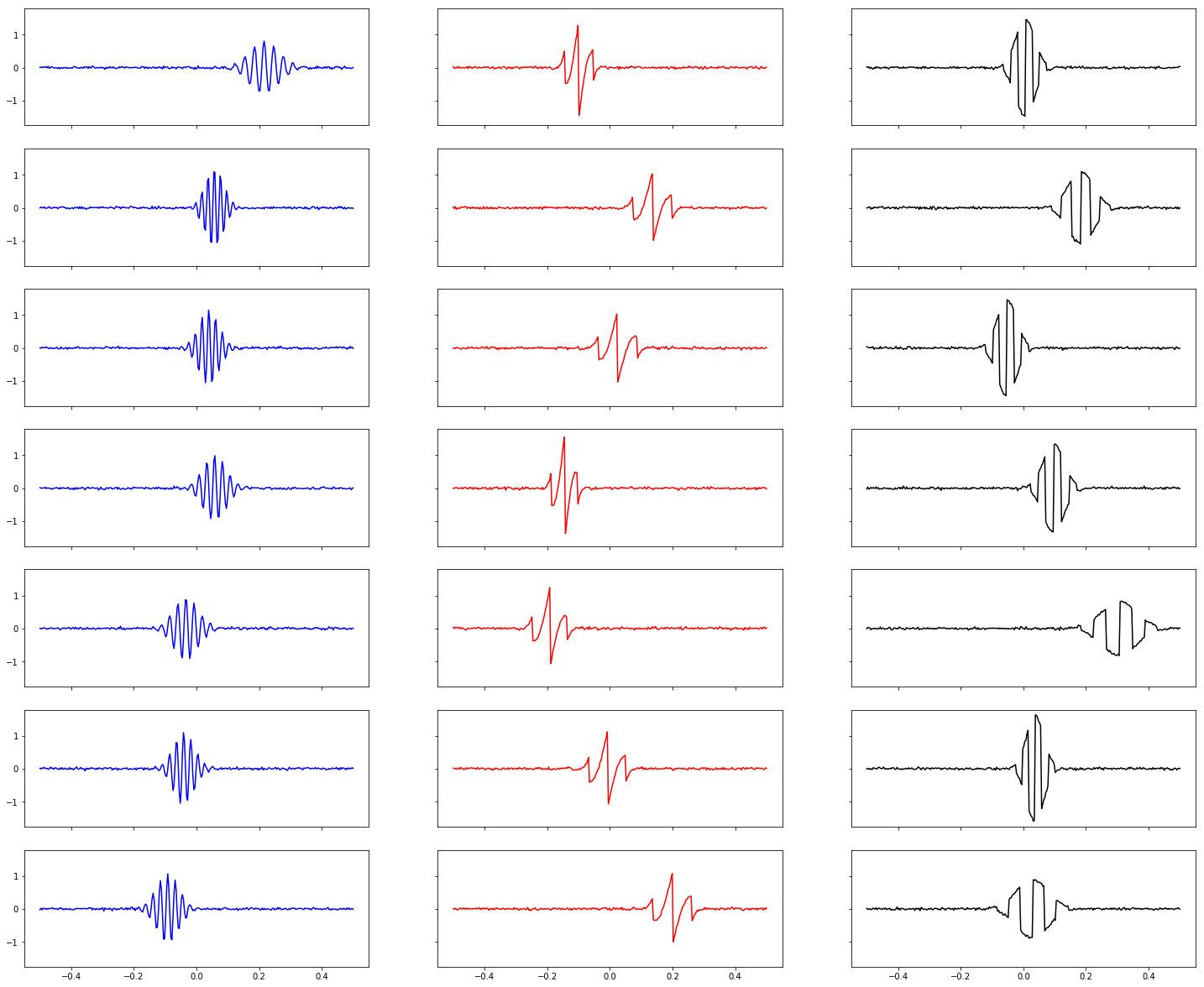}
	\caption{Three signal classes:  A Gabor wave, an apodized sawtooth wave, and an apodized square wave are randomly translated and scaled  to form three signal classes. Seven example training signals are shown per class. }
\label{fig:signals_for_classification}
\end{figure}

\begin{figure}[H]
	\centering
	\includegraphics[width=1.0\linewidth]{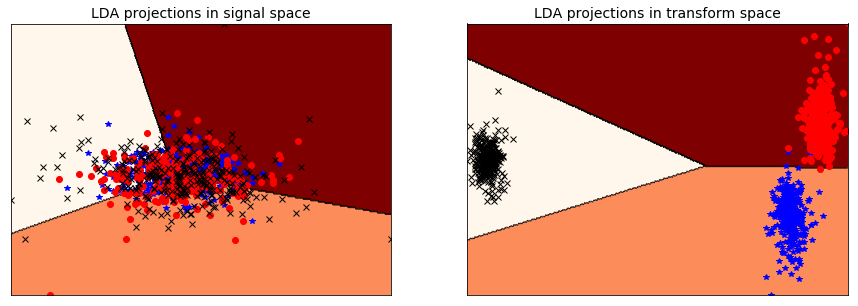}
	\caption{Classification of test signals (from the three classes depicted in Figure \ref {fig:signals_for_classification}): Projection  to LDA subspace learned from training data. Left panel:    The linear classification method is unsuccessful classifying signal data in its original space.  Right panel: Test data is much better separated in SCDT  space. }
\label{fig:LDA_classification}
\end{figure}

\section{Summary and Conclusions} 
\noindent


This paper extends the Cumulative Distribution Transform \cite{park2018cumulative} to signed measures of arbitrary mass, permitting the application of the technique to arbitrary one-dimensional signals. This extension significantly broadens the number of potential applications of the transform. The idea is based on viewing 1-D signals (measured data) as measures, and matching the measure corresponding to the signal to be transformed to a chosen reference measure. This matching is obtained using a \emph{push-forward} of the reference measure by a function derived from the cumulations of the reference and  signal measures. The operator that produces the push-forward function is what we call the Signed Cumulative Ditribution Transform (SCDT).  Signed measures are handled using the Jordan decomposition, where the positive and negative portions are handled separately and independently. Theorem \ref{bijection_signed_measure} shows that the mapping is bijective from the space of signed finite measures to the transform space. As such, the signed cumulative distribution transform described in this paper can be viewed as a mathematical signal representation method, with analytical forward and inverse operations, for arbitrary 1-D signals.

Following earlier work on the CDT \cite{park2018cumulative}, we also described several properties of the newly introduced SCDT.  Proposition \ref{Composition_Signed} states that for $g(t)$ a strictly increasing surjective function, the SCDT of a signal (measure) $\eta$ satisfying $F_\eta =F_\nu\circ g$ will be related to the SCDT of $\nu$ via $\Big((\eta^\pm)^\star,\|\eta\|\Big)=\Big(g^\inv\circ(\nu^\pm)^\star, \|\nu^\pm\|\Big)$. Simple corollaries of the proposition include the signal translation and scaling properties \ref{coro_translation}, \ref{coro_dilation}, describing that transformations $g(t)$ that shift the signal along the independent variable ($t$ in this case) become transformations that shift the signal in the dependent variable in transform space.

Proposition \ref{Composition_Signed} (composition) and corollaries \ref{coro_translation} (translation) and \ref{coro_dilation} (scaling/dilation) relate to the analysis of signals under rigid and non-rigid deformations (e.g., deformations in the independent variable time or space). Proposition \ref{convexity_Signed} describes a generative model for classes of signals under the presence of deformations and describes the necessary and sufficient conditions for such classes to be convex in SCDT space. Section \ref{sec:metric} describes a metric for 1D signals using the SCDT and Theorem \ref{linearSeparability} utilizes it, in combination with Proposition \ref{convexity_Signed} and the Hahn-Banach separation theorem to establish sufficient conditions for linear separability of such signal classes in SCDT spaces. Finally, a computational example application of the technique to classifying signals under random translation and dilations using a simple linear classifier is shown. 

The definition of the SCDT given in equations \eqref{one reference} and \eqref{IMU} is one of several possibilities. In the definition used here, we consider a positive, non trivial, reference measure $\mu_0$ to which the positive and negative components of the Jordan decomposition of the signal measure are matched. Another possibility would be to use a reference measure that also admits a decomposition $\mu_0 = \mu_0^+ - \mu_0^-$, and then replace equation \eqref{one reference} with a similar version that matches the corresponding positive and negative parts:    
\begin{equation}
        \TMP_{\mu_0}(\nu) = \left(\TMP_{\mu_0^+}(\nu^+), \TMP_{\mu_0^-}(\nu^-)\right) \notag
    \end{equation}
with the range of the transform being
 \begin{equation} 
        \mathcal{I}_{\mu_0}:=\{(f,r,g,s)\in \STPM^2:  \,     f_\#\mu_0^+ \perp g_\#\mu_0^- \, \text { for } r,s\in\R^+\}. \notag
    \end{equation}
The properties of the transform (e.g. invertibility, composition, translation, dilation, convexity, etc.) under this alternative definition would remain the same, although the proofs would be slightly altered. Naturally, this alternative definition would require both $\mu_0^+$ and $\mu_0^-$ to be non-trivial. 

In summary, this paper presents a new tool for representing arbitrary signals by matching their corresponding measures to a reference measure. As such, it enables the extraction and analysis of information related to rigid and non-rigid deformations of the signal, which are difficult to decode especially in nonlinear estimation and classification problems. Future work will include exploring the application of SCDT to a variety of signal estimation and classification problems, extension of the transform presented here to higher dimensional signals, as well as sampling, reconstruction, compression, and approximation problems.

\section{Proofs of results}
This section contains the proofs of our results. Some of the proofs rely on certain properties of the monotone generalized inverse whose properties and proofs are relegated to the appendix. 
\subsection{Proofs for Section \ref {TFM}}\

We start with several essential lemmas. The proof of Lemma \ref{Lema_aux1} can be found in \cite[Lemma 2.4]{santambrogio2015optimal} and also in the proof of \cite[Corollary 2.2]{thorpeNotes}. The proof of Lemmas \ref{lemma_pushforward} and \ref{OTNOtes+Ambrosio} combine the proofs of   \cite[Theorem 3.1]{ambrosio2003}, \cite[Theorem 2.5]{santambrogio2015optimal} and \cite[Corollary 2.2] {thorpeNotes}. We include them here for completeness, and because we extend them to $\RE$. In what follows, we will use $\mu_{Leb}|_{[0,1]}$ to denote the Lebesgue measure on $[0,1]$.

\begin{lemma}\label{Lema_aux1}
Let $\mu$ be a probability measure that does not give mass to atoms, and let  $F_{\mu}$ be its  CDF (see \eqref {CD}). Then $(F_{\mu})_{\#}\mu = \mu_{Leb}|_{[0,1]}$. As a consequence, for every $l \in [0,1],$ the set $I_l=\{x: F_{\mu}(x) = l\}$ is $\mu$ negligible. 
\end{lemma}
\begin{proof}
Note that $F_{\mu}$ is continuous, as $\mu$ does not give mass to atoms. So, for all $y \in [0,1]$ the set $F_{\mu}^{-1}([0,y])$ is closed, in particular $F_{\mu}^{-1}([0,y]) = [-\infty,x_y],$ for some $x_y\in\RE$ with $F_{\mu}(x_y) = y.$ Now, for $y \in [0,1],$
    \begin{align*}
        (F_{\mu})_{\#}\mu([0,y]) &= \mu((F_{\mu})^{-1}[0,y]) \\ 
         &= F_{\mu}(x_y) = y=\mu_{Leb}|_{[0,1]}([0,y]).
    \end{align*}
    Hence, $(F_{\mu})_{\#} \mu = \mu_{Leb}|_{[0,1]} $ as Borel measures, since they coincide on every interval $[0,a]$, with $a\in[0,1]$. 
    As a consequence, for all $l\in [0,1]$,
    \begin{equation*}
        \mu\left(\{x: F_{\mu}(x) = l\}\right)=\mu(F_{\mu}^{-1}(\{l\}))=(F_{\mu})_{\#}\mu(\{l\})= \mu_{Leb}|_{[0,1]}(\{l\})=0.
    \end{equation*} 
\end{proof}

\begin{lemma}\label{lemma_pushforward}  Let $\mu, \nu$ be two probability measures on $\RE$, and assume that   $\mu$ does not give mass to atoms. Then,
 $\psi = F_\nu^\inv\circ F_{\mu}$ is non-decreasing and satisfies  $\psi_\# \mu = \nu$.
\end{lemma}

\begin{proof}
By Proposition \ref {prop_inv}(i), the function $F_\nu^\inv$ is non-decreasing. Thus,   $\psi: \RE\to\RE$ is non-decreasing since it is a composition of two non-decreasing functions.
To show $\psi_\# \mu = \nu$, we prove  that 
    $$\nu=(F_\nu^\inv)_\# \mu_{Leb}|_{[0,1]}.$$
Then using $(F_{\mu})_\# \mu=\mu_{Leb}|_{[0,1]}$ from Lemma \ref{Lema_aux1}, and properties of the push-forward operator, we obtain the desired result.
Indeed, 
    for $y \in \RE,$ using Proposition \ref{geninv&Lebesgue2} 
    we get, 
    \begin{align*}
        (F_\nu^\inv)_\#\mu_{Leb} |_{[0,1]}([-\infty,y]) &= \mu_{Leb} |_{[0,1]}\left((F_{\nu}^\inv)^{-1}([-\infty,y])\right) \\
         &= \mu_{Leb} |_{[0,1]}\left(\{x\in \RE :   F_{\nu}^\inv(x) \leq y\}\right) \\
        &= \mu_{Leb} |_{[0,1]}\left(\{x\in\RE : \, x\leq F_{\nu}(y) \}\right) \\
        &= F_{\nu}(y). 
    \end{align*}
\end{proof}

\begin{lemma} \label{OTNOtes+Ambrosio}  Let $\mu, \nu$ be two probability measures on $\RE$, and assume that  $\mu$  does not give mass to atoms. If $\phi:\RE\to \RE$ is a  non-decreasing function such that $\nu=\phi_\# \mu$, then  $\phi=F_\nu^\inv\circ F_{\mu}$ $\mu$-$a.e$.
\end{lemma}

\begin{proof}
Let $\phi$ be any non-decreasing function such that $\nu = \phi_\# \mu$ and assume, (possibly modifying $\phi$ on a countable set) that $\phi$ is right-continuous. Let 
    \begin{equation*}
        T:=\{s\in \supp \mu: (s,s')\cap \supp \mu = \emptyset \text{ for some } s'>s\}
    \end{equation*}
    and notice that $T$ is at most countable (since we can index $T$ with a family of pairwise disjoint open intervals). We claim $\phi\geq F_\nu^\inv\circ F_{\mu}$ on $\supp \mu - T$. Indeed,  for $s\in \supp \mu - T$ and $s'>s$ we have 
    \begin{equation*}
        \nu([-\infty,\phi(s')]) = \mu(\phi^{-1}([-\infty,\phi(s')])) \geq \mu([-\infty,s'])>\mu([-\infty,s]),
    \end{equation*}
    where last inequality follows from the fact that $(s,s')\cap \supp \mu \neq \emptyset$. Thus, we have that $F_\nu(\phi(s')) > F_{\mu}(s)$. Using the fact that $F_\nu$ and $F_\mu$ are non-decreasing and right continuous, we can apply  Proposition \ref {basic}(i) to get  $\phi(s')\geq F_\nu^\inv\circ F_{\mu}(s), a.e.$-$\mu$. Taking $s'\to s$ we obtain $\phi(s)\geq F_\nu^\inv\circ F_{\mu}(s)$ for every $s\in \supp \mu - T$. In particular, using Lemma \ref{lemma_pushforward} we get
     \begin{equation*}
         \int (\phi - F_\nu^\inv \circ F_{\mu})(x) \ d\mu(x) = 
         \int \phi(x) \ d\mu(x) - \int (F_\nu^\inv \circ F_{\mu})(x) \ d\mu(x) =
         \int y \ d\nu(y) - \int y \ d\nu(y) = 0.
     \end{equation*}
    But $\phi-F_\nu^\inv \circ F_{\mu}\ge 0, \; a.e.-\mu$. Therefore, $\phi = F_\nu^\inv \circ F_{\mu}$ $\mu$-a.e.
    
\end{proof}

\subsubsection{Proofs of Theorem \ref{Prob_bijection} and Corollary \ref{prob_restrict}.}
\begin{proof}[Proof of Theorem \ref{Prob_bijection}]
 Injectivity holds since, by  Lemma \ref{OTNOtes+Ambrosio},  if $\widehat{\nu}=\widehat{\eta} $   $\mu_0$-a.e., then $$\nu=\widehat{\nu}_\# \mu_0=\widehat{\eta}_\# \mu_0=\eta.$$
 For surjectivity, if $\phi:\RE\to\RE$ is a non-decreasing $\mu_0$-a.e function, 
 then the push-forward $\nu := \phi_\#\mu_0$ of $\mu_0$  by $\phi$ is a probability measure on $\RE$. By Lemma \ref{OTNOtes+Ambrosio}, for any probability measure $\nu$, the transform $\widehat{\nu} = F_\nu^\inv\circ F_{\mu_0}$  is a unique non-decreasing $\mu_0-a.e.$ function which satisfies $\widehat{\nu}_\# \mu_0=\nu$. Therefore, $\phi=\widehat{\nu}$ $\mu_0$-a.e., i.e. $\phi$ lies in the image of the transform.
\end{proof}

\begin{proof}[Proof of Corollary \ref{prob_restrict}]
Note that $\mathcal{P}(\Omega)\subseteq \mathcal{P}(\RE)$,  $\ND(\Omega)\subseteq \ND$, and that $\TMP_{\mu_0}:\mathcal{P}(\RE)\to \ND$ is a bijection. Thus, the restriction $\TMP_{\mu_0}|_{\mathcal{P}(\Omega)}$ is trivially one to one. \par
The surjectivity of $\TMP_{\mu_0}|_{\mathcal{P}(\Omega)}:\mathcal{P}(\Omega)\rightarrow \ND(\Omega)$ follows from the surjectivity of $\TMP_{\mu_0}$ if we can show that $\TMP_{\mu_0}^{-1}(\ND(\Omega))\subset \mathcal{P}(\Omega)$.

Indeed, given $\varphi\in\ND(\Omega)$, since $\varphi(x)\in\Omega$ $\mu_0$-a.e., we have that for each Borel set $E\subseteq \RE$, 
\begin{align*}
    \varphi_\#\mu_0(E)&=\mu_0\Big(\varphi^{-1}(E)\Big)\\
    &=\mu_0(\{x: \, \varphi(x)\in E\})\\
    &=\mu_0(\{x: \, \varphi(x)\in E\cap\Omega\})=
    \varphi_\#\mu_0(E\cap\Omega).
\end{align*}
Thus,  $(\TMP_{\mu_0})^{-1}(\varphi)=\varphi_\#\mu_0 \in \mathcal{P}(\Omega)$. 
\end{proof}


\subsubsection{Proofs of Theorem \ref{bij_non-nomnalized} and Corollary \ref{FMsub}}
\begin{proof}[Proof of Theorem \ref{bij_non-nomnalized}]
Note that if a function $f:\RE\to\RE$ is a  non-decreasing  $\mu_0$-a.e., it is also  non-decreasing  $\left(\frac{1}{\|\mu_0\|}\mu_0\right)$-a.e.

In order to prove injectivity, let $\nu$ and $\eta$ be two non-zero, finite positive measures on $\RE$ such that $\nu^*=\eta^*$ and $\|\nu\|=\|\eta\|$. 
Hence, by Theorem \ref{Prob_bijection},  $\frac{1}{\|\nu\|}\nu =\frac{1}{\|\eta\|}\eta$ and 
we obtain $\nu=\eta$.  If $\|\nu\|=0$  and $\|\eta\|\ne 0, $ then from the Definition \ref {PMT}, $\widehat \nu\ne \widehat \eta.$

In order to prove surjectivity, consider the pair $(f,r)$ where $f:\RE\to\RE$ is a function non-decreasing  $\mu_0$-a.e. and $r$ is a positive real number. Let $\nu:=rf_\# \frac{\mu_0}{\|\mu_0\|}$, then $\nu$ is a finite positive Borel measure with $\|\nu\|=r$ since $f_\# \big(\frac{\mu_0}{\|\mu_0\|}\big)$ is a probability measure. In addition,  
$$\nu^*=F_{\frac{\nu}{\|\nu\|}}^\inv\circ F_{\frac{\mu_0}{\|\mu_0\|}}$$
and by Lemma \ref{OTNOtes+Ambrosio} it is the unique non-decreasing $\mu_0$-a.e. function that satisfies 
$$\nu^*_\# \big( \frac{\mu_0}{\|\mu_0\|} \big) =  \frac{1}{r}\nu.$$ 
Thus, $f=\nu^*$ $\mu_0$-a.e., and  the inverse transform is given by $$(f,r)\mapsto rf_\# \big(\frac{\mu_0}{\|\mu_0\|}\big). $$ 

If $(f,r)=(0,0)$, then, from Definition \ref {PMT}, it is the transform of the zero measure. 
\end{proof}

\begin{proof}[Proof of Corollary \ref{FMsub}] Corollary \ref{FMsub} follows from Theorem \ref{bij_non-nomnalized} as in the proof of Corollary \ref{prob_restrict}. 

\end{proof}

\subsubsection{Proofs of Theorem \ref{bijection_signed_measure} and Corollary \ref{corollary_signed_measure}} 

\begin{proof} [Proofs of Theorem \ref{bijection_signed_measure}]Given a signed measure $\nu$, $\TMP_{\mu_0}(\nu)$ is well-defined since there exists only one pair of mutually singular positive measures, $\nu^+$ and $\nu^-$, such that $\nu=\nu^+-\nu^-$.

If $\TMP_{\mu_0}(\nu)=\TMP_{\mu_0}(\eta)$, then $\TMP_{\mu_0}(\nu^+)=\TMP_{\mu_0}(\eta^+)$ and $\TMP_{\mu_0}(\nu^-)=\TMP_{\mu_0}(\eta^-)$ (where the operators  $\TMP_{\mu_0}$ should be understood from the context). Thus, injectivity follows by applying Theorem \ref{Prob_bijection} separately on the positive and negative part.

In addition, given $\nu\in\SM(\RE)$, $\TMP_{\mu_0}(\nu)\in\mathcal{I}_{\mu_0}$ since
\begin{equation*}
    (\nu^+)^*_\#\mu=\frac{\nu^+}{\|\nu^+\|} \perp(\nu^-)^*_\#\mu=\frac{\nu^-}{\|\nu^-\|}.
\end{equation*}
In order to prove the surjectivity, consider  $(f,r,g,s)\in\mathcal{I}_{\mu_0}$. Define 
\begin{equation}\label{inv_sign}
\nu:=rf_\# \Big(\frac{\mu_0}{\|\mu_0\|}\Big) - sg_\# \Big(\frac{\mu_0}{\|\mu_0\|}\Big)\in \SM(\RE).    
\end{equation}
For $r,s\in\R^+$, since $f_\#\mu_0 \perp g_\#\mu_0$, then 
$$ rf_\# \Big(\frac{\mu_0}{\|\mu_0\|}\Big)\perp sg_\# \Big(\frac{\mu_0}{|\mu_0|}\Big).$$
Both measures, $rf_\# \Big(\frac{\mu_0}{\|\mu_0\|}\Big)$ and $sg_\# \Big(\frac{\mu_0}{\|\mu_0\|}\Big)$, are positive measures and therefore they are a Jordan decomposition for $\nu$, that is 
$$\nu^+=rf_\# \Big(\frac{\mu_0}{\|\mu_0\|}\Big) \quad  \text{ and } \quad \nu^-=sg_\# \Big(\frac{\mu_0}{\|\mu_0\|}\Big).$$
In particular,  $\nu^+(\RE)=r$ and $\nu^-(\RE)=s$. By Theorem  \ref{Prob_bijection}, 
\begin{equation*}
    f= \left(\frac{1}{r}\nu^+\right)^* \qquad \text{ and }  \qquad g= \left(\frac{1}{s}\nu^-\right)^*.
\end{equation*}
Thus, $\TMP_{\mu_0}(\nu)=(f,r,g,s)$,
and it is clear that  the inverse transform is given by \eqref{inv_sign}.
Finally, if $s=0$ and $r>0$ (resp. $r=0$ and $s>0$) the proof reduces to the one of Theorem \ref{bij_non-nomnalized}, and the point $(0,0,0,0)$ is reached by the zero measure by definition.
\end{proof}

\begin{proof}[Proof of Corollary \ref{corollary_signed_measure}]
The proof of Corollary \ref{corollary_signed_measure} follows from Theorem \ref{bijection_signed_measure} as in Corollary \ref{FMsub}. 
\end{proof}

\subsection{Proofs of Section \ref {POF}}
The following Lemma is useful for the proofs of this section.
\begin{lemma} \label{Decomposition}
Let $\nu,\eta$ be two signed measures on $\RE$ such that $F_{\eta}(x) = F_{\nu}(g(x)),$ for some strictly increasing bijection $g: \RE \rightarrow \RE.$ If $\{\Omega_{\nu^+},\Omega_{\nu^-}\}$ is a Hahn decomposition of $\nu$, then the set made by the pre-images $\{g^{-1}(\Omega_{\nu^+}),g^{-1}(\Omega_{\nu^-})\}$ is a Hahn decomposition of $\eta.$ Thus, $F_{\eta^\pm}(x) = F_{\nu^\pm}(g(x)).$ 
\end{lemma}

\begin{proof}
 Throughout this proof, we will use the fact that if $\nu$ and $\eta$ satisfy the hypothesis of the Lemma \ref{Decomposition}, then, $F_{\eta}(x) = F_{\nu}(g(x)) $ if and only if $g_{\#} \eta = \nu.$ In particular, under the hypothesis of the Lemma, $\eta(\RE)=\nu(\RE).$
 
 Clearly, $\RE = g^{-1}(\Omega_{\nu^+})\cup g^{-1}(\Omega_{\nu^-})$ and $g^{-1}(\Omega_{\nu^+}) \cap g^{-1}(\Omega_{\nu^-}) = \phi.$  Now we show that, $g^{-1}(\Omega_{\nu^+})$ is a positive set for $\eta,$ i.e. $\eta(E) \ge 0,$ for all $\eta$ - measurable sets $E \subset g^{-1}(\Omega_{\nu^+}).$
Since $E \subset g^{-1}(\Omega_{\nu^+})$ and $g$ is a bijection, there exists $\nu$ - measurable $H \subset \Omega_{\nu^+},$ such that $E = g^{-1}(H).$ Thus, $\eta(E) = \eta(g^{-1}(H)) = \nu(H) \ge 0 .$ Analogously, we can show $g^{-1}(\Omega_{\nu^-})$ is a negative set for $\eta$.\\
Now, we show $F_{\eta^+}(x) = F_{\nu^+}(g(x)).$ A similar argument would follow for its negative counterpart.
\begin{align*}
    F_{\eta^+}(x) &= \eta^+([-\infty,x]) \\
    &= \eta([-\infty,x] \cap g^{-1}(\Omega_{\nu^+})) \\
    &= \eta(g^{-1} \circ g ([-\infty,x]) \cap g^{-1}(\Omega_{\nu^+})) \\
    &= \eta(g^{-1} ( g ([-\infty,x]) \cap \Omega_{\nu^+})) \\
    &= \nu(g([-\infty,x]) \cap \Omega_{\nu^+}) \\
    &= \nu([-\infty,g(x)] \cap \Omega_{\nu^+}) \\
    &= \nu^+ ([-\infty,g(x)]) = F_{\nu^+}(g(x)).
\end{align*}
\end{proof}
\begin{proof}[Proof of Proposition \ref{Composition_Signed}] Using the relation, $F_{\eta}(x)=F_{\nu}(g(x))$, we see that $F_{\nu}\circ g$ is a cumulation.  Also, by definition of $\eta $ and $\nu$, we have $F_{\eta^+}(x)-F_{\eta^-}(x)=(F_{\nu}\circ g)_{+}(x)-(F_{\nu}\circ g)_{-}(x)$ and by Lemma \ref{Decomposition}, we get, $F_{\eta^+}(x) = F_{\nu^+}(g(x))$ and $F_{\eta^-}(x) = F_{\nu^-}(g(x)).$ In particular, $\|\eta^\pm\|=\|\nu^\pm\|$.     
Since g is strictly increasing, and using Proposition \ref {prop_inv}, the SCDT of $\eta$ with respect to $\mu_0$ when $\|\nu^+\|\ne 0$ and $\|\nu^-\|\ne 0$ is given by 

\begin{align*}
    \TMP_{\mu_0}( \eta) &=\left(F_{\frac {\eta^+}{\|\eta^+\|}}^\inv\circ F_{\frac {\mu_0}{\|\mu_0\|}}, \|\eta^+\|, F_{\frac {\eta^-}{\|\eta^-\|}}^\inv\circ F_{\frac {\mu_0}{\|\mu_0\|}},\|\eta^-\|\right) \\ &= \left(\big(F_{\frac {\nu^+}{\|\nu^+\|}}\circ g\big)^\inv\circ F_{\frac {\mu_0}{\|\mu_0\|}} , \|\eta^+\| , \big(F_{\frac {\nu^+}{\|\nu^+\|}}\circ g\big)^\inv\circ F_{\frac {\mu_0}{\|\mu_0\|}},\|\eta^-\|\right) \\ &=
    \left(g^\inv \circ F^\inv_{\frac {\nu^+}{\|\nu^+\|}}\circ  F_{\frac {\mu_0}{\|\mu_0\|}}, \|\nu^+\|,g^\inv \circ F^\inv_{\frac {\nu^-}{\|\nu^-\|}}\circ  F_{\frac {\mu_0}{\|\mu_0\|}},\|\nu^-\|\right) \\ &=
    \left(g^\inv \circ ( \nu^+)^*,\|\nu^+\|,g^\inv \circ (\nu^-)^*,\|\nu^-\|\right).
\end{align*}
The cases when $\|\nu^+\|=0$ or $\|\nu^-\|=0$ can be done a similar fashion.

\end{proof} 

We adopt a proof similar to the one in \cite{aldroubi2020partitioning}.
\begin{proof}[Proof of Proposition \ref{convexity_Signed}] 
Given $\nu\in\SM(\RE)$ we consider 
$\mathcal S_{\nu,\mathbb {H}}$ as is \eqref{S_nu,H}. By the definition of $\mathcal S_{\nu,\mathbb {H}}$, Proposition \ref{Composition_Signed} and Theorem \ref{bijection_signed_measure}
we have that 
\begin{equation}\label{S_hat}
  \omega \in \widehat{\mathcal S_{\nu,\mathbb {H}}} \Leftrightarrow \omega= (h^{-1}\circ (\nu^+)^*,\|\nu^+\|,h^{-1}\circ (\nu^-)^*, \|\nu^-\|) \, \text{ for some } h\in \mathbb {H}.  
\end{equation}

Assume that $\mathbb{H}^{-1}$ is convex and fix $\nu\in \SM(\RE)$. Let $\eta_h,\eta_g$ be two arbitrary elements in $S_{\nu,\mathbb{H}}$, that is, they are defined by  $F_{\eta_h}=F_\nu\circ h$ and 
$F_{\eta_g}=F_\nu\circ g$ for 
$h,g\in \mathbb{H}$ (here we are using the characterization of measures according Proposition \ref{characterization_meausres_on_RE} from the Appendix). For any $\alpha\in[0,1]$, applying Proposition \ref{Composition_Signed} we have
\[
\begin{split}
 \alpha ({\eta}^\pm_h)^*+(1-\alpha)( {\eta}_g^\pm)^*=&\alpha \,  h^{-1}\circ ({\nu}^\pm)^*+(1-\alpha) \,  g^{-1}\circ ({\nu}^\pm)^*\\
 =&\big(\alpha h^{-1} +(1-\alpha) g^{-1}\big)\circ ( {\nu}^\pm)^*.
\end{split}
\]
In addition $\alpha \|\nu^\pm\|+(1-\alpha)\|\nu^\pm\|=\|\nu^\pm\|$. Thus, $\widehat {\mathcal S_{\nu,\mathbb{H} }}$ is convex. 

For the converse statement, let $h^{-1},g^{-1}\in \mathbb{H}^{-1}$ and $\alpha\in[0,1]$. For $\nu\in \SM(\RE)$, assuming that $\widehat {\mathcal S_{\nu,\mathbb{H} }}$ is convex we have that 
\[\Big(\big(\alpha h^{-1} +(1-\alpha) g^{-1}\big)\circ ( {\nu}^+)^*,\|\nu^+\|,\big(\alpha h^{-1} +(1-\alpha) g^{-1}\big)\circ ({\nu}^-)^*, \|\nu^-\|\Big)\in \widehat {\mathcal S_{\nu,\mathbb{H}} }
\]
Thus, by the characterization of $\widehat {\mathcal S_{\nu,\mathbb{H}}}$ given  by \eqref{S_hat} we obtain that $\alpha h^{-1} +(1-\alpha) g^{-1}$ coincides with a function in $\mathbb{H}^{-1}$ on the range of $({\nu^{\pm}})^*$. 
Taking 
the family of target measures $\{\nu_T\}_{T\in \R}\cup\{\delta_{\infty}, \delta_{-\infty}\}$, where $\delta_{\pm\infty}$ are Dirac measures centered at $\pm \infty$, and $\nu_T$ is defined by 
$$\nu_T(E):=\int_{E}\chi_{[0,1]}(x-T) \, d\mu_{Leb}(x) \qquad \forall E\subset \RE \text{ Borel subset},$$
and assuming $\widehat {\mathcal S_{\nu_T,\mathbb{H} }}$ convex for every $T\in\R$, 
by Corollary \ref{coro_dilation}
we can conclude that $\alpha h^{-1} +(1-\alpha) g^{-1}\in \mathbb {H}^{-1}$. Thus $\mathbb {H}^{-1}$ is convex. 
\end{proof}


Proposition \ref{prop_wass} and its proof are well-known \cite{thorpeNotes}. We include them in this paper for readability and completeness.  
\begin{proof}[Proof of Proposition \ref{prop_wass}]
It is well known that (cf. \cite{santambrogio2015optimal,villani2003topics})
\begin{equation}\label{wass_2}
  d_{W^2}(\nu,\eta)=\left(\int_{0}^1 |F^\inv_\nu(t)-F^\inv_{\eta}(t)|^2 \, dt\right)^{1/2} . 
\end{equation}
Then,  by a change  of variables and using  Lemma \ref{Lema_aux1}, we obtain 
\begin{align*}
    \|\TMP_{\mu_0}(\nu)  -\TMP_{\mu_0}(\eta) \|_{L^2(\mu_0)}^2 &=\int_\R |F_\nu^\inv\circ F_{\mu_0}-F_{\eta}^\inv\circ F_{\mu_0}|^2 \, d\mu_0 \notag\\
    &=\int_0^1 |F_\nu^\inv(t)-F_{\eta}^\inv(t)|^2 dt\notag\\
    &=\left(d_{W^2}(\nu,\eta)\right)^2 
\end{align*}
\end{proof}

 \begin{proof}[Proof of Theorem \ref {linearSeparability}]
Since $\widehat {\mathcal D}_\nu, \widehat{\mathcal D}_\eta$ are finite  subsets of a normed space $\Big(L^2(\mu_0)\times\R\Big)^2$, their convex hulls $\text{conv}\big(\widehat {\mathcal D}_\nu\big)$, and $\text{conv}\big(\widehat {\mathcal D}_\eta\big)$ are compact.  In addition, since 
$\widehat {\mathcal D}_\nu, \widehat{\mathcal D}_\eta$ are subsets of the convex sets $\widehat{\mathcal S_{\nu,\mathbb{ H}}}$ and $ \widehat{\mathcal S_{\eta,\mathbb{ H}}}$, then $\text{conv}\big(\widehat {\mathcal D}_\nu\big)$,   $\text{conv}\big(\widehat {\mathcal D}_\eta\big)$ are also subsets of $\widehat{\mathcal S_{\nu,\mathbb{ H}}}$ and $ \widehat{\mathcal S_{\eta,\mathbb{ H}}}$. Finally, since $\mathcal S_{\nu,\mathbb{ H}}$ and $ \mathcal S_{\eta,\mathbb{ H}}$ are disjoint and $\TMP_{\mu_0}$ is one to one, we get that  $\text{conv}\big(\widehat {\mathcal D}_\nu\big)$, and $\text{conv}\big(\widehat {\mathcal D}_\eta\big)$ are disjoint non-empty convex and compact sets. By the Hahn-Banach Separation Theorem, they can be separated by a linear functional $f$. In particular, $f$ separates $\widehat {\mathcal D}_\nu, \widehat{\mathcal D}_\eta$.
\end{proof}

\section{Appendix}\label{appendix}

\subsection{Measures on $\RE$}

We recall that we are using the standard topology on $\RE$, which is given by the standard topology on $\R$ and a base of open neighborhoods of $\infty$ is $\{ (a,\infty)\cup\{\infty\}\}_{a\in \R}$ and analogously for the point $-\infty$.
Then, the Borel $\sigma$-algebra of $\RE$ is generated by $\{-\infty\}$, $\{\infty\}$, $(-\infty,a]$ with $a\in\R$. 


\begin{proposition}\label{characterization_meausres_on_RE}
 Let $\mu$ be a finite positive Borel measure on $\RE$, then the cumulation of $\mu$ 
 defined by $F_{\mu}(x):=\mu([-\infty,x])$ for all $x\in\RE$,
 is non-decreasing,  right-continuous, and
    \begin{equation} \label{LimF}
        F_{\mu}(-\infty)  \geq 0 \qquad \text{ and } 
       \qquad F_{\mu}(\infty)<\infty.
    \end{equation}
    Conversely, for any right-continuous, non-decreasing function $F$ on $\RE$ satisfying \eqref{LimF}, there is a unique finite positive Borel measure $\mu$ on $\RE$ such that 
    \begin{equation*}
        F(x)= \mu([-\infty,x])\text{ for all } x\in\RE.
    \end{equation*}
\end{proposition}

This is an extension of the so called Lebesgue-Stieltjes Measure.

\begin{proof}\
    Direct part: 
    \begin{itemize}
        \item $F_{\mu}$ is non-decreasing because given $x,y\in \RE$ with $x\leq y$, since $[-\infty,x]\subseteq[-\infty,y]$,
        $$F_{\mu}(x)=\mu([-\infty,x])\leq \mu([-\infty,y])=F_{\mu}(y).$$
        \item $F_{\mu}(-\infty)=\mu(\{-\infty\})\geq 0$ (since $\mu$ is positive).
        \item Let $x\in\R$ 
        $$F_{\mu}(x)=\mu([-\infty,x])=\lim_{n\to\infty}\mu([-\infty,x+\frac{1}{n}])=\lim_{n\to\infty}F_{\mu}(x+\frac{1}{n})$$
        so $F_{\mu}$ is right-continuous on $\R$. Indeed, $F_{\mu}$ is right-continuous on $\RE$: If $x=\infty$ there is nothing to prove. If $x=-\infty$, then $x+a=-\infty$ for all $a\in\R$ therefore, $\displaystyle \lim_{z\to-\infty} F_\mu(z)=F_\mu(-\infty)$.
    \end{itemize}
    
    For the converse part, let $F:\RE\to\RE$ be a right-continuous, non-decreasing function satisfying \eqref{LimF}.
    Denoting
    \begin{equation*}
        r=F(-\infty)  \qquad \text{ and } \qquad  s= F(\infty),
    \end{equation*}
    define
    \begin{gather*}
        T:\RE\to\RE\\
        T(\alpha):=\inf\{x: F(x)\geq\alpha\}
    \end{gather*}
    Then, $T$ is non-decreasing, and therefore it is a measurable function.
    Notice that for each $x\in\RE$
    $$\mu_{{Leb}_{|_{[0,s]}}}\left(T^{-1}([-\infty,x])\right)=F(x)$$
    In particular, if $x=\infty$, $T^{-1}([-\infty,\infty])\cap[0,s]=\RE\cap[0,s]=[0,s]$ and if $x=-\infty,$ $T^{-1}(\{-\infty\})\cap[0,s]=[0,r]$ 
    because $T^{-1}(\{-\infty\})=\{y: \, \inf\{x: \,  F(x)\geq y\}=-\infty\}=[-\infty,r]$ since the range of $F$ is a subset of $[r,s]$.
    
    We define
    $$\mu:=T_\#\mu_{{Leb}_{|_{[0,s]}}}$$
    and, for each $x\in\RE$, we obtain
    \begin{equation*}
        \mu([-\infty,x])=F(x).
    \end{equation*}

    Notice that
    \begin{equation*}
     \mu(\{\infty\})=s-\lim_{x\to \infty} F(x) \qquad \text{ and } \qquad \mu(\{-\infty\})=r=\lim_{x\to -\infty} F(x) 
    \end{equation*}
    The uniqueness of the measure $\mu$ is a consequence of the Carath\'eodory Extension Theorem, which asserts that any finite measure on an algebra $\mathcal{A}$ extends in a unique
    way to a measure on the $\sigma$-algebra generated by $\mathcal{A}$. Indeed,
    the equation 
        $$\mu((a,b])=F(b)-F(a) \qquad \forall a,b\in\RE$$
    implies that there is only one extension to the algebra of sets generated by $\{-\infty\}$ and  half-open intervals $(a,b]$ with $a,b\in\RE$, and so there is only one extension to the $\sigma$-algebra generated by these sets, which is $\mathcal{B}(\RE)$.
\end{proof}

\subsection{Monotone Generalised Inverse} 
In this section, we introduce the monotone generalized inverse for functions defined on the extended real line $\RE$, and we provide some of its relevant properties. In particular, the monotone generalized inverse of any function is always  non-decreasing, and if a function $F$ is continuous and strictly increasing then its monotone generalized inverse $F^\inv$ and its standard inverse $F^{-1}$ coincide. This inverse and its properties are essential in defining and studying the transport transform on $\RE$. It has already been introduced in connection to transport theory but only for functions on the real line \cite{thorpeNotes}, and  some its properties are well-known \cite{Embrechts:13}. However, we also need other properties that we derive in this appendix. 

\begin{definition}
 For a function $F:\RE \to \RE$, the monotone generalized inverse of $F$ is the function $F^{\inv}:\RE\to\RE$  defined as
    \begin{equation*}
        F^{\inv}(y)  = \inf\{x \in \RE: F(x)>y\}.
    \end{equation*}
    In particular (since $\inf \emptyset = \infty$),  $F^{\inv}(\infty) = \infty.$
\end{definition}
The following properties of $F^{\inv}$ are used in this paper.
\begin{remark} \label{RemFMI} For any function $F:\RE \to \RE$ 
\begin{itemize}
    \item $F^\inv(\infty)=\infty$
    \item If $F^\inv(y)=\infty$, then, $\{s\in \R: F(s)>y\}=\emptyset$.
    \item  If  $F$ is non-decreasing function on $\RE$, then $F^\inv(y)= -\infty$, then $\{s\in \R: F(s)>y\}=\R$.
    \item  If $F$ is non-decreasing function on $\RE$, then \begin{equation*}
        F^{\inv}(y)  = \inf\{x \in \R: F(x)>y\}=\inf\{x \in \RE: F(x)>y\}.
    \end{equation*}
\end{itemize}
\end{remark}
\begin{proposition} \label{basic}
For any function $F:\RE \to \RE$ 
    \begin{itemize}
        \item[(i)] $F(x)>y \Rightarrow x\geq F^\inv(y).$ \label{eq1}
         \item[(ii)] $x<F^\inv(y)  \Rightarrow    F(x)\leq y.$ \label{eq2}
    \end{itemize} 
\end{proposition}
\begin{proof}
   Since $F(x)>y,$ (i) follows from the definition of $F^\inv,$ and
   (ii) is the contrapositive of (i).

\end{proof}

\begin{proposition} \label{non-dec}
If $F : \RE \to \RE$  is non-decreasing then
    \begin{itemize}
        \item[(i)] $F(x)=y \Rightarrow x\leq F^\inv(y) \qquad 
        (\text{in particular } \quad  x\leq F^\inv(F(x))).$ \label{i} 
        \item[(ii)] $F(x)<y \Rightarrow x \leq F^\inv(y).$ \label{eq4} 
        \item[(iii)] $x>F^\inv(y)  \Rightarrow F(x)>y  .$ 
    \end{itemize}
\end{proposition}
\begin{proof}
\noindent
\begin{itemize}
    \item[(i)] Since $F(x)=y$ and $F$ is non-decreasing, $x<s$ for any $s\in U=\{s:F(s)>y\}$. Thus, (i) follows from definition of $F^\inv.$ 
    \item[(ii)]  If $F(x)<y$. Then, since $F$ is non-decreasing,  $x<s$ for any $s$ such that $F(s)>y$. Thus, $x\le F^\inv(y).$   
     \item[(iii)]  From (i) and (ii) imply that if $F(x)\le y$, then $x\le F^\inv(y),$ which is the contrapositive of (iii).

\end{itemize}
\end{proof}

\begin{proposition} \label{non-dec&rightcts}
If $F : \R \cup \{-\infty\}\to \RE $ is right-continuous and non-decreasing then, 
\begin{itemize}
        \item[(i)] $x=F^\inv(y)  \Rightarrow  F(x)\geq y 
        \qquad (\text{in particular } \quad  y\leq F(F^\inv(y)))$
        \item[(ii)] $F(x)<y \Rightarrow x<F^\inv(y).$ \label{eq6}
        \item[(iii)] $x\geq F^\inv(y)  \Rightarrow F(x)\geq y.$ \label{eq7}
    \end{itemize} 
\end{proposition}

\begin{proof}
\noindent
\begin{itemize}
    \item[(i)] If $x=F^\inv(y)\in\R$, by definition of infimum, for all $\varepsilon>0$, there exists $s_\varepsilon$ such that $F(s_\varepsilon)>y$ and $s_\varepsilon< F^\inv(y)+\varepsilon=x+\varepsilon$. Since, $F$ is non-decreasing and right-continuous, we have $y\leq F(x)$. If $x=F^\inv(y)= -\infty$, then $\{s\in \R: F(s)>y\}\supset \R$. Therefore, by right continuity at $-\infty$, $F(-\infty)\geq y$.   To see that (i) is not valid if we allow $\RE$ to be the domain of $F$, consider any non-decreasing function $F$ that satisfies $F(\infty)=1$. Then we have that $F^\inv(2)=\infty$. However, $F(\infty)=1<2$. 
    \item[(ii)] If $F(x)<y$  then by \ref{non-dec}(ii), $x\leq F^\inv(y)$. If $x= F^\inv(y)$ then by (i), $F(x)\geq y$ which contradicts $F(x)<y$.  
    \item[(iii)]  This implication is the contrapositive of (ii) above. 
\end{itemize}
\end{proof}

\begin{proposition} \label{strict_inc}
If $F : \RE \to \RE$ is strictly increasing then
    \begin{align}
        F(x)=y \Rightarrow x=F^\inv(y) \qquad (\text{in particular } \quad  x= F^\inv(F(x)))\label{eq8}
    \end{align}
\end{proposition}
\begin{proof}
If $F(x)=y$ then      $\{s\in\RE:F(s)>y\} = \{s\in\RE:s>x\}$.
Taking infimum on both sides we get $F^\inv(y)=x$.
\end{proof}

\begin{proposition}\label{prop_inv} For a function $F : \RE \to \RE,$ the following hold:
    \begin{enumerate}[(i)]
        \item \label{inv_rcont} $F^\inv$ is right-continuous and non-decreasing. 
        \item   $(F^\inv)^\inv=F$ if and only if $F$ is non-decreasing and right-continuous and $F(\infty)=\infty$.
        In other words,  $F$ is non-decreasing and right-continuous if and only if $F(x)=\inf\{t: F^\inv(t)>x\}$.
        \item \label{comp_leq} If $F,G:\RE\to\RE$ are non-decreasing, then $(F\circ G)^\inv\leq G^\inv\circ F^\inv$.
        \item \label{comp_inv}Let  $F,G:\RE\to\RE$ be non-decreasing. If 
        $G$ is strictly increasing,
        then $(F\circ G)^\inv= G^\inv\circ F^\inv$. 
        \label{eq12}
    \end{enumerate}
\end{proposition}
\begin{proof}

\begin{enumerate}[(i)]
    \item Let $y_1,y_2$ be such that $y_1<y_2$. Then 
    \[
        \{x:F(x)>y_2\}\subseteq \{x:F(x)>y_1.\}
    \]
    Taking the infimum on both sides we get that $F^\inv$ is non-decreasing.
    
    Let $\epsilon >0$ be given. For $\Delta>0$, set $U_{\Delta}:=\{x : F(x)> y+\Delta \}$ and $U:=\{x : F(x)> y\}$. Since $F^\inv$ is non-decreasing, if $\Delta_1<\Delta_2$, then $F^\inv(t+\Delta_1)\le F^\inv(t+\Delta_2)$. 
    Choose $x\in U$ such that $x-F^\inv(y)<\epsilon$, and $\delta>0$ so small that $x\in U_\delta$. Then for any $\Delta<\delta$, we have
    \[
        F^\inv(y)\le F^\inv(y+\Delta)\le F^\inv(y+\delta)\le x.
    \]
    Thus, 
    \[
        0\le F^\inv(y+\Delta)-F^\inv(y)\le F^\inv(y+\delta)-F^\inv(y)\le x-F^\inv(y)<\epsilon .
    \]
Hence $\lim\limits_{\Delta\to 0+}F^\inv(y+\Delta)=F^\inv(y)$.

    \item If $(F^\inv)^\inv=F$, then,  by Part (i), $F$ is non-decreasing and right-continuous.  In addition by Remark \ref {RemFMI}, $F(\infty)=\infty$.

    Conversely, let $x\in \RE$. Then using Proposition \ref {basic} (ii) we get set inclusion
    \[
    \{t: F^\inv(t)>x\}\subseteq \{t:F(x)\le t\}.
    \]
    Thus,
    \[
    (F^\inv)^\inv(x)=\inf \{t: F^\inv(t)>x\} \ge \inf \{t:F(x)\le t\}=F(x).
    \]
   To prove that $(F^\inv)^\inv(x)\le F(x)$ for all $x\in \RE$, first assume that there exists $x_0 \in \RE-\{\infty\}$ such that $(F^\inv)^\inv(x_0) > F(x_0)$. For such an $x_0$ it must be that $(F^\inv)^\inv(x_0)>-\infty$. Let $y_0\in\R$ be such that
   $F(x_0)<y_0<(F^\inv)^\inv(x_0)$. Since $F$ is right-continuous, there exists $\epsilon>0$ such that $F(x) <y_0$ for all $x\in [x_0,x_0+\epsilon)$. In addition since $F$ is non-decreasing and right-continuous, it follows from Proposition \ref {non-dec&rightcts} (ii) that $x <F^\inv(y_0)$ for all $x\in [x_0,x_0+\epsilon)$.
   This last inequality, by definition of $(F^\inv)^\inv$, implies that $(F^\inv)^\inv(x)\le y_0.$ But  $(F^\inv)^\inv(x_0)>y_0$ leading to a contradiction.  Finally, since $(F^\inv)^\inv(\infty)=\infty$, and by assumption, $F(\infty)=\infty$  we get part (ii).
   
\item Since $F$ and $G$ are non-decreasing, from Proposition \ref {non-dec}(iii) we see that, 
    $$\{x:x>G^\inv(F^\inv(y))\} \subseteq \{x: G(x)>F^\inv(y)\} \subseteq \{x:F(G(x))>y\}.$$
    Therefore, 
    \begin{align*}
        (F \circ G)^\inv(y) &= \inf \{x:F(G(x))>y\}\\ &\leq \inf\{x:x>G^\inv(F^\inv(y))\} = G^\inv(F^\inv(y)).
    \end{align*} 
\item  If $G$ is strictly increasing, by \eqref{eq8} and  Proposition \ref{basic} 
    \begin{align*}
        \{x:F(G(x))>y\} \subseteq \{x:G(x)\geq F^\inv(y)\} \subseteq \{x:x\geq G^\inv(F^\inv(y))\}.
    \end{align*}
    Therefore, in any case, 
    \begin{align*}
        (F \circ G)^\inv(y) &= \inf \{x:F(G(x))>y\}\\ &\geq \inf\{x:x\geq G^\inv(F^\inv(y))\} = G^\inv(F^\inv(y)),
    \end{align*} 
    Using this and (\ref{comp_leq}) above, (since $F$ and $G$ are non-decreasing) we conclude $(F \circ G)^\inv=G^\inv\circ F^\inv$.
 \end{enumerate}
\end{proof}
 \begin{proposition} \label{geninv&Lebesgue2}
Let $F:\RE\to \RE$ be cumulative distribution of a probability measure $\nu$ on $\RE$. Then, for each $y \in \RE$
$$\mu_{Leb}|_{[0,1]}\Big(\{x\in \RE:F^\inv(x) \leq y\}\Big)
 =\mu_{Leb}|_{[0,1]}\Big(\{x\in\RE: x\leq F(y)\}\Big),$$
 where $\mu_{Leb}$ is the Lebesgue measure. 
\end{proposition}
 
 \begin{proof} Since  $F$ is cumulative distribution of a probability measure $\nu$ on $\RE$, it is non-decreasing and right continuous. We then note that 
 \[
  \{x\in \RE:F^\inv(x) \leq y\} \subset \{\infty\} \cup \{x\in \R\cup \{-\infty\}:F^\inv(x) \leq y\} \\
 \]
 Thus, 
 \begin{equation} \label{SetEq}
 \mu_{Leb}|_{[0,1]}\Big(\{x\in \RE:F^\inv(x) \leq y\}\Big)=\mu_{Leb}|_{[0,1]}\Big(\{x\in \R\cup \{-\infty\}:F^\inv(x) \leq y\}\Big).
 \end{equation}
 
 By Proposition \ref {non-dec&rightcts} (iii) we get

\begin{align*}
  \{x\in \R\cup \{-\infty\}:F^\inv(x) \leq y\} &\subset \{x\in \R\cup \{-\infty\}: x\leq F(y)\} \\
 &=\{x\in \R\cup \{-\infty\}: x< F(y)\}\cup  \{x\in \R\cup \{-\infty\}: x= F(y)\}.  
 \end{align*}
 Thus, 
 \begin{equation} \label{Fineq}
 \mu_{Leb}|_{[0,1]}\Big(\{x\in \R\cup \{-\infty\}:F^\inv(x) \leq y\}\Big)\le \mu_{Leb}|_{[0,1]}\Big(\{x \in  \R\cup \{-\infty\}: x< F(y)\}\Big)=F(y)
 \end{equation}
 On the other hand, using Proposition \ref {basic} (i), we have 
 \[
 \{x \in \RE: x< F(y)\}\subset \{x\in \RE: F^\inv(x) \leq y\}
 \]
 Hence, using \eqref {SetEq} 
 \begin{equation} \label{Sineq}
 \begin{split}
 F(y)=\mu_{Leb}|_{[0,1]}\Big(\{x \in \RE: x< F(y)\}\Big)&\le \mu_{Leb}|_{[0,1]}\Big(\{x\in \RE:F^\inv(x)\leq y\}\Big)\\
 &=\mu_{Leb}|_{[0,1]}\Big( \{x\in \R\cup \{-\infty\}:F^\inv(x)\leq y\}\Big)
 \end{split}
 \end{equation}
 and the conclusion of the proposition follows from inequalities \eqref{Fineq} and \eqref{Sineq}.
\end{proof}

\printbibliography

\end{document}